\documentclass[11pt]{article}
\usepackage{color}
\usepackage[normalem]{ulem}
\usepackage{tikz}
\usepackage{todonotes}
\usepackage{epsfig,enumerate,float,amsmath,amsfonts,amsbsy,amssymb,amsthm,mathrsfs,lineno,ifpdf,relsize,float,etoolbox}
\usepackage{xcolor}
\usepackage[colorlinks=true, allcolors=blue]{hyperref}
\usepackage[ruled,vlined,linesnumbered,procnumbered]{algorithm2e}
\usepackage[mathscr]{euscript}
\usepackage{graphicx}
\usepackage{setspace}
\usepackage{makecell}
\usepackage[utf8]{inputenc}
\usepackage[english]{babel}
\usepackage[numbers]{natbib}
\usepackage[usenames,dvipsnames]{pstricks}
\usepackage{pst-grad}
\usepackage{float}
\usepackage{tikz}
\usetikzlibrary{positioning,fit,backgrounds,arrows.meta}
\theoremstyle{plain}
\newtheorem{theorem}{Theorem}
\newtheorem{lemma}{Lemma}

\newtheorem{coro}{Corollary}
\newtheorem{obs}{Observation}

\newtheorem{claim}{Claim}

\newtheorem{defi}{Definition}
\newtheorem{construction}{Construction}

\makeatletter
\AtBeginEnvironment{procedure}{\let\c@algocf\c@procedure}
\makeatother

\let\oldenumerate\enumerate
\renewcommand{\enumerate}{
	\oldenumerate
	\setlength{\itemsep}{1.5pt}
	\setlength{\parskip}{0pt}
	\setlength{\parsep}{0pt}
}
\title{On the Classical and Parameterized Complexity of Strong Odd Coloring}
\author{$^1$Dinabandhu Pradhan, $^1$Vaishali Sharma, $^2$Shaily Verma\\ \\
	$^{1}$Department of Mathematics \& Computing\\ Indian Institute of Technology (ISM), Dhanbad\\
	$^{2}$Department of Mathematics\\ Indian Institute of Technology, Jodhpur, India\\
	\small \tt Email: dina@iitism.ac.in; vaishali.sharma7791@gmail.com; shailyverma@iitj.ac.in}

\begin{document}

	\maketitle
	\begin{abstract}
		A \emph{strong odd $k$-coloring} of a graph $G$ is a proper $k$-coloring such that every color appearing in the neighborhood of a non-isolated vertex appears an odd number of times. The minimum $k$ for which $G$ admits a strong odd $k$-coloring is the \emph{strong odd chromatic number}, denoted by $\chi_{\text{so}}(G)$, of $G$. Given a graph $G$ and an integer $k$, \textsc{strong odd $k$-colorability} problem asks whether $G$ admits a strong odd $k$-coloring.
		
		It is known that \textsc{strong odd $k$-colorability} is NP-complete in general graphs. In this paper, we prove that the problem is NP-complete on perfect elimination bipartite graphs for $k\geq3$, which is a subclass of bipartite graphs. Furthermore, we show that $\chi_{\text{so}}(G)$ is inapproximable within a factor of $O(n^{\frac{1}{2}-\varepsilon})$ for every $\varepsilon>0$. On the positive side, we obtain a linear time algorithm to compute an optimal strong odd coloring for block graphs. From a parameterized perspective, we present an FPT algorithm for \textsc{strong odd $k$-colorability} when parameterized by treewidth. Moreover, we show that the problem cannot be solved in time $(k-\varepsilon)^{\texttt{tw}}n^{O(1)}$ for every $k\geq3$ and $\varepsilon>0$ when parameterized by treewidth under SETH. Furthermore, we show that \textsc{strong odd $k$-colorability} does not admit a polynomial kernel when parameterized by feedback vertex set. Lastly, we prove that \textsc{strong odd $k$-colorability} is W[1]-hard when parameterized by clique-width.
	\end{abstract}
	

\section{Introduction}

Petru\v{s}evski and \v{S}krekovski~\cite{PS1} introduced the notion of odd coloring in 2022. A proper coloring $f$ of a graph $G$ is an \emph{odd coloring} if, for every non-isolated vertex $v\in V(G)$, at least one color appears an odd number of times in $N(v)$. The minimum $k$ for which $G$ admits an odd $k$-coloring is the \emph{odd chromatic number} of $G$ and is denoted by $\chi_{o}(G)$. Odd coloring has been studied extensively since then, see e.g.,~\cite{ACCCHHHKZ,CPS_odd,Cho2023,KMOOT,Kashima2026}. Bhyravarapu et al. studied parameterized complexity of odd coloring~\cite{Bhyravarapu2026}.

Kwon and Park~\cite{Kwon2026} introduced a stronger version of odd coloring with a stronger parity requirement. A proper coloring $f$ of a graph $G$ is a \emph{strong odd coloring} if, for every non-isolated vertex $v\in V(G)$, every color that occurs in $N(v)$ occurs an odd number of times in $N(v)$. The minimum $k$ for which $G$ admits a strong odd $k$-coloring is the \emph{strong odd chromatic number} of $G$, denoted by $\chi_{\text{so}}(G)$. Consequently, $\chi(G)\leq \chi_{o}(G)\leq \chi_{\text{so}}(G)\leq \chi(G^2)$, where $G^2$ is the square of $G$, obtained by adding an edge between every pair of vertices at distance two in $G$.

The difference between $\chi_{o}(G)$, $\chi_{\text{so}}(G)$, and $\chi(G^2)$ can be arbitrarily large. For the complete bipartite graph $G= K_{m,n}$, we have $\chi(G^2)=m+n$, whereas $\chi_{\text{so}}(G)\leq 4$. Thus, the difference between $\chi_\text{so}(G)$ and $\chi(G^2)$ is unbounded. Moreover, $\chi_{\text{so}}(G)-\chi_{o}(G)$ is unbounded. For example, Kwon and Park~\cite{Kwon2026} showed that for $G=K_n\square K_n$, $\chi_{\text{so}}(G)=\frac{(\Delta(G)+2)^2}{4}$, while the general bound $\chi_{o}(G)\leq 2\Delta(G)+1$~\cite{CPS_odd} gives a linear bound in $\Delta$. These examples highlight the distinct behavior of strong odd coloring and motivate its structural and algorithmic study.

The study of strong odd coloring has primarily focused on the structural bounds and determination of the strong odd chromatic number for specific graph classes. However, its computational complexity remains largely unexplored. Kwon and Park~\cite{Kwon2026} studied strong odd coloring for graphs with bounded maximum average degree, where \emph{maximum average degree} of a graph $G$ is the maximum of $\frac{2|E(H)|}{|V(H)|}$ over all non-empty subgraphs $H$ of $G$. Caro et al.~\cite{Caro2025} established that the strong odd chromatic number of planar graphs is bounded by a constant, and posed a question whether the strong odd chromatic number of all planar graphs is at most $12$. Pang et al.~\cite{Pang2026} answered this question in the negative by constructing a planar graph $G$ with $\chi_{\text{so}}(G)=13$. Later, Manattu et al.~\cite{Manattu2026} constructed a planar graph $G$ with $\chi_{\text{so}}(G)=17$ and also evaluated the strong odd chromatic number of the join of cycles and empty graphs and one point union of graphs. Goetze et al.~\cite{Goetze2025} proved that the strong odd chromatic number is bounded on every proper minor-closed graph class. Recently, Pilipczuk~\cite{Pilipczuk2025} extended this boundedness result by proving that strong odd chromatic number is bounded on graph classes of bounded expansion.

Given a graph $G$ and an integer $k$, the \textsc{strong odd $k$-colorability} problem asks whether $G$ admits a strong odd coloring using at most $k$ colors. The corresponding optimization problem is to determine $\chi_{\text{so}}(G)$, denoted by \textsc{strong odd chromatic number}. Belmonte et al.~\cite{Belmonte2026} proved that \textsc{strong odd $k$-colorability} is NP-complete for $k\geq 3$. Consequently, the problem is para-NP-hard when parameterized by $k$. Caro et al.~\cite{Caro2025} proved that \textsc{strong odd $k$-colorability} is linear time solvable for trees and unicyclic graphs. In this paper, we study the classical complexity and initiate the study of the parameterized complexity of \textsc{strong odd $k$-colorability}.

\vspace{0.3cm}
\noindent\textbf{Our contributions:}

\vspace{-0.2cm}

\begin{itemize}
	\item We show that \textsc{strong odd $k$-colorability} is NP-complete on bipartite graphs for every $k\geq 3$. In particular, we prove that \textsc{strong odd $k$-colorability} is NP-complete on perfect elimination bipartite graphs for every $k\geq 3$, which is a subclass of bipartite graphs. Furthermore, we prove that the strong odd chromatic number cannot be approximated within a factor of $O(n^{\frac{1}{2}-\varepsilon})$ for any $\varepsilon>0$, unless $\mathrm{P}=\mathrm{NP}$.
	
	\item We give a linear time algorithm for computing an optimal strong odd coloring of block graphs. We traverse the cut-tree in breadth-first search order and recursively color the block children of each cut vertex. We reuse colors from the parent to compute an optimal coloring, and ensure that the strong odd condition satisfies.
	
	\item We give an FPT algorithm for \textsc{strong odd $k$-colorability} parameterized by the treewidth. We design a dynamic programming algorithm on a nice tree decomposition of the graph that runs in time $O((k\cdot9^k)^{\texttt{tw}+1}\cdot n^{O(1)})$. Moreover, for every $k\geq 3$ and every $\varepsilon>0$, we show that \textsc{strong odd $k$-colorability} cannot be solved in time $(k-\varepsilon)^{\mathrm{tw}}n^{O(1)}$, unless the Strong Exponential Time Hypothesis (SETH) fails. It remains open whether the running time of our algorithm can be improved to narrow the gap with the lower bound.
	
	\item We show that \textsc{strong odd $k$-colorability} does not admit a polynomial kernel when parameterized by feedback vertex set. Moreover, we prove that the problem is W[1]-hard when parameterized by clique-width.
	
\end{itemize}

The paper is organized as follows. In Section~2, we introduce some necessary definitions. Section~3 presents hardness results. In Section~4, we present a linear time algorithm on block graphs. In Section~5, we study the parameterized complexity of the problem.

\section{Preliminaries}

All graphs considered in this paper are finite, connected, simple, and undirected. We follow the notation and terminology used in \cite{textBook}. For a positive integer $k$, let $[k]=\{1,2,\ldots,k\}$. Let $G$ be a graph. An assignment of colors to a set $S\subseteq V(G)$ is a \emph{partial coloring} of $G$. We say a vertex $v\in V(G)$ satisfies the \emph{strong odd condition} under a coloring if each color that appears in $N(v)$, appears an odd number of times in $N(v)$.

Let $G=(X,Y,E)$ be a bipartite graph. An edge $xy\in E(G)$ is called \emph{bisimplicial} if the subgraph of $G$ induced by $N(x)\cup N(y)$ is a complete bipartite graph. A \emph{perfect edge elimination scheme} of a graph $G$ is an ordering of edges $(x_1y_1, x_2y_2,\ldots,x_ky_k)$ such that each edge $x_{j+1}y_{j+1}$ is bisimplicial in the graph $G[V(G)\setminus\{x_1, y_1, x_2, y_2, \ldots, x_{j}, y_{j}\}]$. A graph is a \emph{perfect elimination bipartite graph} if it admits a perfect edge elimination scheme.

Let $G$ be a graph. A vertex $v$ of $G$ is a cut vertex if $G[V\setminus\{v\}]$ is disconnected. A \emph{block} of $G$ is a maximal connected subgraph that has no cut vertex. A graph $G$ is a \emph{block graph} if every block of $G$ is a clique.

A parameterized problem is a problem whose instances are pairs $(x,d)$, where $x$ is the input and $d\in\mathbb{N}$ is the parameter. A parameterized problem $Q$ is \emph{fixed-parameter tractable} (FPT) if there exists a computable function $g$ such that $Q$ can be solved in time $g(d)|x|^{O(1)}$. A \emph{kernelization} for a parameterized problem $Q$ is a polynomial time
algorithm that, given an instance $(I,k)$, outputs an equivalent instance $(I',k')$ such that $|I'|+k' \leq g(k)$, where $g$ is a computable function. If $g$ is a polynomial function, then the problem is said to admit a \emph{polynomial kernel}. We use the standard notions of parameterized reductions and W[1]-hardness. For further details on parameterized complexity, we refer the reader to \cite{Cygan2015}.

\begin{defi}[Tree Decomposition]
	A  \emph{tree decomposition} of a graph $G$ is a pair $\mathcal{T}=(T,\{X_t\}_{t\in V(T)})$, where $T$ is
	a tree whose every node $t$ is assigned a vertex subset
	$X_t\subseteq V(G)$, called a \emph{bag}, such that the following conditions hold:
	\begin{enumerate}
		\item[(T1)] $\displaystyle\bigcup_{t\in V(T)}X_t=V(G)$. In other words, every vertex of $G$ is in at least one bag.
		
		\item[(T2)] For every edge $uv\in E(G)$, there exists a node $t$ of $T$ such that bag $X_t$ contains both $u$ and $v$.
		
		\item[(T3)] For every $u\in V(G)$, the set $T_u=\{t\in V(T):u\in X_t\}$, i.e., the set of nodes whose corresponding bags contain $u$, induces a connected subtree of $T$.
	\end{enumerate}
\end{defi}

The \emph{width} of a tree decomposition $\mathcal{T}=(T,\{X_t\}_{t\in V(T)})$ equals $\max_{t\in V(T)}|X_t|-1$. The \emph{treewidth} of a graph $G$, denoted by
$\texttt{tw}(G)$, is the minimum possible width of a tree
decomposition of $G$.

A rooted tree decomposition $(T,\{X_t\}_{t\in V(T)})$ is \emph{nice} if (i) $X_r$ and $X_\ell$ both have exactly one vertex for every leaf $\ell$ of $T$ and (ii) Every non-leaf node of $T$ is of one of the following three types:
\vspace{-0.2cm}
\begin{itemize}
	\item \textbf{Introduce node:} a node $t$ with exactly one child $t'$ such that $X_t=X_{t'}\cup\{v\}$ for some vertex $v\notin X_{t'}$; we say that $v$ is \emph{introduced} at $t$.
	
	\item \textbf{Forget node:} a node $t$ with exactly one child $t'$ such that $X_t=X_{t'}\setminus\{w\}$ for some vertex $w\in X_{t'}$; we say that $w$ is \emph{forgotten} at $t$.
	
	\item \textbf{Join node:} a node $t$ with two children
	$t_1$ and $t_2$ such that $X_t=X_{t_1}=X_{t_2}$.
\end{itemize}

 \begin{lemma}[\cite{Cygan2015}]
	 If a graph $G$ admits a tree decomposition of width at most $k$, then it also admits a nice tree decomposition of width at most $k$. Moreover, given a tree decomposition
	 $\mathcal{T}=(T,\{X_t\}_{t\in V(T)})$ of $G$ of width at most $k$, one can compute a nice tree decomposition of $G$ of width at most $k$ that has at most $O(k|V(G)|)$ nodes in time $O\left(k^2\cdot\max(|V(T)|,|V(G)|)\right)$.
\end{lemma}

Let $G$ be a graph. A set $X\subseteq V(G)$ is a \emph{vertex cover} of $G$ if every edge of $G$ has at least one endpoint in $X$. The minimum cardinality of a vertex cover of $G$ is called the \emph{vertex cover number} of $G$ and is denoted by $\texttt{vc}(G)$. A set $X\subseteq V(G)$ is a \emph{feedback vertex set} of $G$ if $G-X$ is a forest. The minimum cardinality of a feedback vertex set of $G$ is the \emph{feedback vertex set number} and is denoted by
$\texttt{fvs}(G)$.

\begin{defi}[Clique-width]
The \emph{clique-width} of a graph $G$, denoted by $\texttt{cw}(G)$, is the minimum number of labels needed to construct $G$ using the following four operations:
\begin{enumerate}
\item \textbf{Introducing a vertex.}
$\Phi=v(i)$ creates a new vertex $v$ with label $i$. Thus, $G_{\Phi}$ consists of a single vertex $v$ with label $i$.

\item \textbf{Disjoint union.}
$\Phi=\Phi'\oplus\Phi''$ constructs the disjoint union of the labeled graphs $G_{\Phi'}$ and $G_{\Phi''}$.

\item \textbf{Introducing edges.}
$\Phi=\eta_{i,j}(\Phi')$ adds an edge between every vertex with label $i$ and every vertex with label $j$, where $i\neq j$, in $G_{\Phi'}$.

\item \textbf{Renaming labels.}
$\Phi=\rho_{i\rightarrow j}(\Phi')$ changes the label of every vertex with label $i$ to label $j$ in $G_{\Phi'}$.
\end{enumerate} 

An expression constructed using these four operations and using at most $w$ labels is called a \emph{$w$-expression}. In other words, the clique-width of a graph $G$ is the minimum integer $w$ for which there exists a $w$-expression that constructs $G$.
\end{defi}

\section{Hardness results}

In this section, we show that \textsc{strong odd $k$-colorability} is NP-complete on perfect elimination bipartite graphs. Further, we show that $\chi_\text{so}(G)$ cannot be approximated within a factor of $O(n^{\frac{1}{2}-\varepsilon})$ for every $\varepsilon>0$.

\begin{theorem}\label{pebg_thm}
	\textsc{strong odd $k$-colorability} is NP-complete for $k\geq3$ for perfect elimination bipartite graphs.
\end{theorem}
\begin{proof}
	It is easy to verify that \textsc{strong odd $k$-colorability} is in NP. Now we show that \textsc{strong odd $k$-colorability} is NP-hard by reducing an instance $(G, k)$ of $k$-\textsc{colorability} to an instance $(H,k)$ of \textsc{strong odd $k$-colorability} such that $G$ is $k$-colorable if and only if $H$ is strong odd $k$-colorable. 
	
	\begin{construction}\label{construction_PEBG}
		Let $G$ be a graph with $V(G) = \{v_1, v_2, \ldots, v_n\}$ and $E(G) = \{e_1, e_2, \ldots, e_m\}$. We construct a graph $H$ from $G$ in polynomial time using the following steps.
		
		\vspace{-0.2cm}
		
		\begin{enumerate}
			\item Take a copy vertex $v_i'$ of each vertex $v_i\in V(G)$. Let $V'=\{v_1', v_2',\ldots,v_n'\}$.
			
			\item Take a vertex $v_{e_i}$ corresponding to each edge $e_i\in E(G)$ and make $v_{e_i}$ adjacent to the copy of both the endpoints of $e_i$ in $V'$. Let $V_{E}^1 = \{v_{e_1}, v_{e_2}, \ldots, v_{e_m}\}$. Clearly, each $v_{e_i}$ is a vertex of degree $2$.
			
			\item For each vertex $v_{e_i}\in V_{E}^1$, introduce a new vertex $v_{e_i}'$. Let $V_{E}^2 = \{v_{e_1}', v_{e_2}', \ldots, v_{e_m}'\}$. We add an edge between $v_{e_i}'\in V_{E}^2$ and $v_j\in V'$ whenever there is an edge between $v_{e_i}$ and $v_j$. Therefore, each $v_{e_i}'\in V_{E}^2$ is a vertex of degree $2$.
			
			\item Attach $k-1$ number of pendant vertices to each vertex $v_i\in V'$, denote them by $\{v_1^i, v^i_2, \ldots v^i_{k-1}\}$. Let $P=\{v^1_1, v^1_2, \ldots v^1_{k-1}, v^2_1, v^2_2, \ldots v^2_{k-1}, \ldots, v^n_1, v^n_2, \ldots v^n_{k-1}\}$.
		\end{enumerate}
		
	\end{construction}

	\begin{figure}[h]
		\begin{center}
			\includegraphics[width=0.5\linewidth]{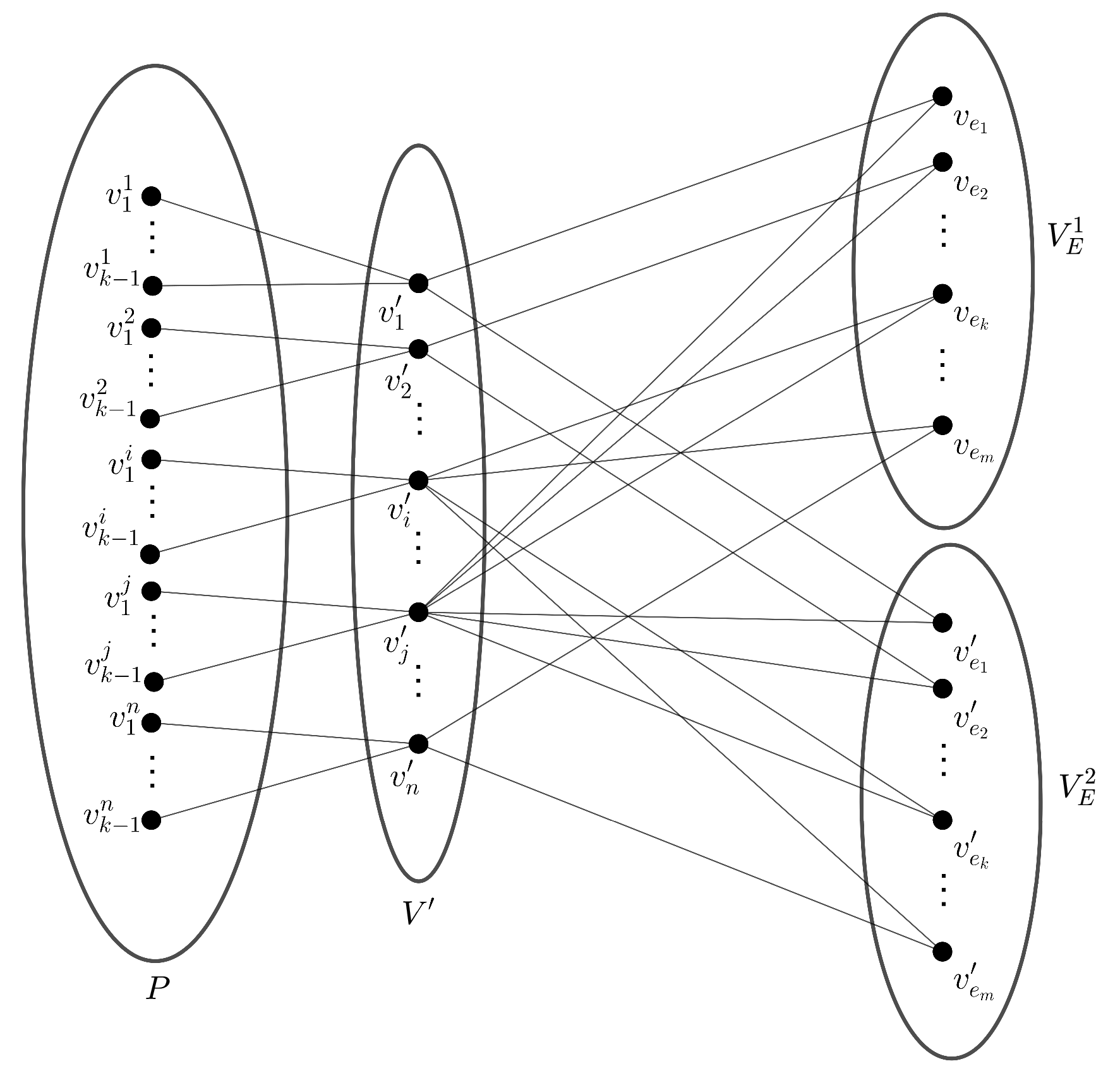}
			\caption{Illustration of Construction~\ref{construction_PEBG}.}
			\label{fig_pebg}
		\end{center}	
	\end{figure}

	Note that $G$ is a bipartite graph. Let $\sigma = (v_1v^1_1, v_2v^2_1, \ldots, v_nv^n_1)$ and $S_i = \{v_1, v^1_1, v_2, v^2_1, \\\ldots, v_i, v^i_1\}$. Every edge $v_{i+1} v^{i+1}_1$ is bisimplicial in $H[V(H) \setminus S_i]$ for each $i\in [n]$. Moreover, there are no edges in the graph $H[V(H) \setminus S_n]$. Therefore, $\sigma$ is a perfect edge elimination scheme of the graph $H$. Thus, $H$ is a perfect elimination bipartite graph.

	\begin{claim}\label{SO_npc}
		$G$ is $k$-colorable if and only if $H$ is strong odd $k$-colorable.
	\end{claim}
	\begin{proof}
	Let $\varphi$ be a proper $k$-coloring of $G$. We extend this coloring to a strong odd $k$-coloring of $H$, say $\varphi^*$. For each $v_i'\in V'$, $\varphi^*(v_i') = \varphi(v_i)$. Color each vertex $v_{e_i}\in V_{E}^1$ by avoiding the colors assigned to both of its neighbors in $V'$. Then color each vertex $v_{e_i}'\in V_E^2$ by assigning the same color as assigned to $v_{e_i}$. Assign distinct colors to each vertex of the set  $\{v^i_1, v^i_2, \ldots, v^i_{k-1}\}$ from the set of colors $[k]\setminus\varphi^*(v_i')$.  
	
	Now, we show that $\varphi^*$ is a strong odd coloring of $H$. Since each vertex $x\in P$ is a pendant vertex, either a color $c$ does not appear at all or it appears an odd number of times in $N(x)$. Let $v_i'$ and $v_j'$ be two neighbors of a vertex $v_{e_i}\in V_E^1$. Since $v_i$ and $v_j$ are adjacent in $G$, $v_i'$ and $v_j'$ are assigned distinct colors under $\varphi^*$. Therefore, either a color $c$ does not appear at all, or it appears an odd number of times in $N(v_{e_i})$ for each $v_{e_i}\in V_{E}^1$. Similarly, either a color $c$ does not appear at all, or it appears an odd number of times in $N(v_{e_i}')$ for each $v_{e_i}'\in V_{E}^2$. The set $V'$ is an independent set in $H$. Thus, neighbors of each $v_i'\in V'$ are only in the sets $V_{E}^1$, $V_{E}^2$, and $P$. Since $v_i'\in V'$ is adjacent to $v_{e_i}'$ whenever $v_i'$ is adjacent to $v_{e_i}$ and $\varphi^*(v_{e_i})=\varphi^*(v_{e_i}')$, each color that appears in $N_{H}(v_i)\cap (V_{E}^1\cup V_{E}^2)$ appears an even number of times. Note that each $c\in [k]\setminus\{\varphi^*(v_i')\}$ appears exactly once in $N(v_i')\cap P$. Consequently, each color of the set $[k]\setminus\{\varphi^*(v_i)\}$ appears an odd number of in $N_{H}(v_i)$. Thus, $\varphi^*$ is a strong odd coloring of $G$.
	
	Conversely, let $\psi$ be a strong odd $k$-coloring of $H$. We restrict the coloring $\psi$ to the vertices of the set $V'$, denoted by $\psi'$. Since $v_i$ and $v_j$ have a common $2$-neighbor whenever there is an edge between $v_i$ and $v_j$ in $G$, $\psi'(v_i)\neq \psi'(v_j)$ for each edge $v_iv_j\in E(G)$. Thus, $\psi'$ is a proper $k$-coloring of $G$.
\end{proof}

It is well known that $k$-\textsc{colorability} is NP-complete, for $k\geq 3$. Therefore, \textsc{strong odd $k$-colorability} is NP-complete for perfect elimination bipartite graphs, $k\geq3$. This completes the proof of Theorem~\ref{pebg_thm}. 
\end{proof}

Now, we provide an inapproximability result for \textsc{strong odd chromatic number}. We use the following result by Zuckerman~\cite{Zuckerman}. 
\begin{theorem}[Zuckerman~\cite{Zuckerman}]\label{chromaticnumber}
	For any $\varepsilon > 0$, \textsc{chromatic number} of a graph cannot be approximated within $O(n^{1-\varepsilon})$ unless P=NP.
\end{theorem}

\begin{theorem}\label{inapproximation}
	For any $\varepsilon>0$, \textsc{strong odd chromatic number} of a graph is inapproximable within $O(n^{\frac{1}{2}-\varepsilon})$ unless P=NP.
\end{theorem}
\begin{proof}
	We construct a graph $H$ from an arbitrary graph $G$ such that a polynomial time $O(n^{\frac{1}{2}-\varepsilon})$-approximation algorithm $\mathcal{A}$ for \textsc{strong odd chromatic number} gives a polynomial time $O(n^{1-\varepsilon})$-approximation algorithm for \textsc{chromatic number}.
	
	\begin{construction}
		Let $G$ be a graph with the vertex set $V(G) = \{v_1, v_2, \ldots, v_n\}$ and the edge set $E(G) = \{e_1, e_2, \ldots, e_m\}$. We construct a graph $H$ from the graph $G$ in polynomial time using the following steps.
		
		\vspace{-0.2cm}
		
		\begin{enumerate}
			
			\item Take two copies of $G$, say $G_1$ and $G_2$. Let $V(G_1) = V_1 = \{u_1, u_2, \ldots, u_n\}$ and $V(G_2) = V_2 = \{u_1', u_2', \ldots, u_n'\}$.

			\item Let $v_i\in V(G)$ be the vertex corresponding to $u_i\in V_1$ and $u_i'\in V_2$. Attach $d(v_i)+n$ pendant vertices to each vertex $u_i\in V_1$ and $u_i'\in V_2$. Let $X_{u_i} = \{u_i^1, u_i^2, \ldots, u_i^{d(v_i)+n}\}$ and
			$X_{u_i'} = \{(u_i')^1, (u_i')^2, \ldots, (u_i')^{d(v_i)+n}\}$
			represent the sets of all the pendant vertices attached to
			$u_i\in V_1$ and $u_i'\in V_2$, respectively.  Let $X = X_{u_1}\cup X_{u_2}\cup\ldots X_{u_n} \cup X_{u'_1}\cup X_{u'_2}\cup\ldots X_{u'_n}$.

			\item Add a vertex $z$ and make it adjacent to each vertex of $V_1\cup V_2\cup X$.
			
			\item Attach $n$ number of pendant vertices to the vertex $z$. Let $S = \{s_1, s_2, \ldots, s_n\}$ be the set of all the pendant vertices attached to the vertex $z$.
		\end{enumerate}
		
	\end{construction}

	\begin{claim} \label{reduction_iff}
		$\chi_{\text{so}}(H)\leq\chi(G)+2$.
	\end{claim}
\begin{proof}
	We prove this by showing that $H$ admits a strong odd coloring with at most $\chi(G)+2$ colors.
	
	Let $\varphi:V(G)\rightarrow [k]$ be an optimal proper coloring of $G$. Note that $\chi(G)=k$. We construct a strong odd coloring $\varphi^*$ of $H$ with at most $k+2$ colors. For each $u_i\in V_1$ and $u_i'\in V_2$, assign $\varphi^*(u_i)=\varphi^*(u_i')=\varphi(v_i)$. Assign $\varphi^*(z)=k+1$. For each $s_i\in S$ with $i\le k$, assign $\varphi^*(s_i)=i$. If $n-k$ is even, then choose a color $\beta\in[k]$ and color each uncolored vertex in $S$ with color $\beta$. If $n-k$ is odd, then choose a color $\beta\in[k]$, color all uncolored vertices in $S\setminus\{s_n\}$ with color $\beta$, and assign color $k+2$ to $s_n$. It remains to color the vertices in $X$. Let $u_i\in V_1$. Recall that $u_i$ corresponds to the vertex $v_i\in V(G)$. We color $\{u_i^1, u_i^2, \ldots, u_i^{d(v_i)}\}$ in $X_{u_i}$ with the same colors as the vertices in $N(v_i)$ i.e., for every color $j\in[k]$, the number of vertices among $\{u_i^1, u_i^2, \ldots, u_i^{d(v_i)}\}$ receiving the color $j$ is equal to the number of vertices in $N_G(v_i)$ receiving the color $j$. We color $\{u_i^{d(v_i)+1}, u_i^{d(v_i)+2}, \ldots, u_i^{d(v_i)+k-1}\}$ in $X_{u_i}$ with $k-1$ colors in $[k]\setminus\{\varphi^*(v)\}$ so that each color is used exactly once. If $n-d(v_i)-k$ is even, then choose a color $\beta\in[k]\setminus\{\varphi^*(v)\}$ and color each vertex in $\{u_i^{d(v_i)+1}, u_i^{d(v_i)+2}, \ldots, u_i^{d(v_i)+k}\}$ with color $\beta$. If $n-d(v_i)-k$ is odd, then choose a color $\beta\in[k]\setminus\{\varphi^*(v)\}$, color all uncolored vertices in $X_v\setminus\{v^{d(v)+n}\}$ with color $\beta$, and assign color $k+2$ to $v^{d(v_i)+n}$. Then color each vertex of $X_{u_i}'$ with the same colors used in $X_{u_i}$, so that the number of vertices colored with color $j$ in $X_{u_i}$ is equal to the number of vertices colored with color $j$ in $X_{u_i'}$.
	
	\textbf{Proper coloring:} We first show that $\varphi^*$ is a proper coloring of $H$. Since $\varphi$ is proper, $\varphi^*$ is proper on $G_1$ and $G_2$. Note that $\varphi^*(z)=k+1$, which is distinct from the color assigned to any of its neighbors. Finally, each vertex $x\in X_v$ receives a color different from $\varphi^*(v)$, and each $s_i\in S$ receives a color different from $\varphi^*(z)$. Thus, $\varphi^*$ is proper.
	
	\textbf{Strong odd condition:} Now we show that each vertex $v\in V(H)$ satisfies the strong odd condition. 
	
	\textbf{Vertices of $V_1\cup V_2$:} First, consider a vertex $v\in V_1\cup V_2$. Fix a color $i\in[k]$ with $i\neq\varphi^*(v)$. By construction, the number of vertices of color $i$ in $(N(v)\cap(V_1\cup V_2))\cup\{v^1, v^2, \ldots, v^{d(v)}\}$ is even. The color $i$ appears an odd number of times among $\{v^{d(v)+1}, v^{d(v)+2}, \ldots, v^{n}\}$. Therefore, color $i$ appears an odd number of times in $N(v)$. Note that the color $k+1$ appears exactly once in $N(v)$, and the color $k+2$ appears at most once in $N(v)$. Thus, every vertex $v\in V_1\cup V_2$ satisfies the strong odd condition.
	
	\textbf{Vertices of $X\cup S$:} Let $x\in X_v$. Then $x$ is adjacent to its corresponding vertex $v\in V_1\cup V_2$ and to $z$. Since $\varphi^\ast(v)\neq\varphi^\ast(z)$, each of these colors appears exactly once in $N(x)$. Therefore, $x$ satisfies the strong odd condition. For every $s_i\in S$, we have $N(s_i)={z}$. Thus, the color $k+1$ appears exactly once in $N(s_i)$, so $s_i$ satisfies the strong odd condition. Therefore, every vertex in $X\cup S$ satisfies the strong odd condition.
	
	\textbf{Vertex $z$:} Finally, consider $z$. We show that every color in $[k]\cup\{k+2\}$ appears an odd number of times in $N(z)$. Fix a color $i\in[k]\cup\{k+2\}$. The number of vertices of color $i$ in $V_1\cup X_{u_1}\cup X_{u_2}\cup \ldots \cup X_{u_n} $ is equal to the number of vertices of color $i$ in $V_2\cup X_{u_1'}\cup X_{u_2'}\cup \ldots \cup X_{u_n'} $. Hence, the total number of vertices of color $i$ in $V_1\cup V_2\cup X$ is even. Note that the color $i$ appears an odd number of times in $S$. Therefore, color $i$ appears an odd number of times in $N(z)$. Hence, $z$ also satisfies the strong odd condition. 
	
	Thus, $\varphi^*$ is a strong odd coloring of $H$ with at most $\chi(G)+2$ colors.
\end{proof}

Assume that there exists a polynomial time algorithm $\mathcal{A}$ that approximates \textsc{strong odd chromatic number} within a factor of $N^{\frac{1}{2}-\varepsilon}$ for a graph with $N$ vertices. Let $\mathcal{A}$ return a strong odd coloring $\psi$ of $H$ with $\alpha$ number of colors. Therefore, 
$\alpha\le N^{\frac{1}{2}-\varepsilon}\cdot\chi_{\text{so}}(H).$
Note that $N=4m+2n^2+3n+1\le6n^2$. By Claim~\ref{reduction_iff}, $\chi_{\text{so}}(H)\le\chi(G)+2$. Since $\chi(G)\ge1$, $\chi(G)+2\le3\cdot\chi(G)$. Thus, 
$\alpha\le (4m+2n^2+3n+1)^{\frac{1}{2}-\varepsilon}\cdot(\chi(G)+2)\le3\cdot(6n^2)^{\frac{1}{2}-\varepsilon}\cdot\chi(G).$
Since $z\in V(H)$ is adjacent to every vertex of $G_1$, $\chi(G)-1\le\chi(H)\le\chi_{\text{so}}(H)$. Since restriction of $\psi$ on $G_1$ is a proper coloring of $G$ with at most $\alpha-1$ colors, $\chi(G)\leq\alpha-1$. We have
$\alpha-1<\alpha\le 3\cdot6^{\frac{1}{2}-\varepsilon}\cdot n^{1-2\varepsilon}\cdot\chi(G)$.
Therefore,

$$\frac{\alpha-1}{\chi(G)}<3\cdot6^{\frac{1}{2}-\varepsilon}\cdot n^{1-2\varepsilon}=O(n^{1-\varepsilon'}),$$

where $\varepsilon'=2\varepsilon$. This is a contradiction to Theorem~\ref{chromaticnumber}. Thus, Theorem~\ref{inapproximation} holds. 
\end{proof}

\section{Linear-time algorithm in block graphs}

In this section, we study strong odd coloring of block graphs. We present a linear time algorithm to compute an optimal strong odd coloring of a block graph $G$. Let $G$ be a block graph. Let $\mathcal{B} = \{B_1, B_2, \ldots, B_k\}$ be the set of all the blocks of $G$ and $\mathcal{C} = \{c_1, c_2, \ldots, c_l\}$ be the set of all the cut vertices of $G$ for some integers $k, l\geq 1$.  We represent $G$ by a tree-like structure known as \emph{cut-tree} (see \cite{AHU_textbook}), denoted by $T(G)$, where $V(T(G)) = \{B_1, B_2, \ldots, B_k, c_1, c_2, \ldots, c_l\}$ and $E(T(G)) = \{(B_i, c_j) :~ c_j \in V(B_i), i \in [k], j \in [l]\}$. The vertices $\{B_1, B_2, \ldots, B_k\}$ are the block vertices of $T(G)$ and the vertices $\{c_1, c_2, \ldots, c_l\}$ are the cut vertices of $T(G)$. The cut-tree $T(G)$ can be obtained in linear time by a depth-first search of $G$. We say a cut-tree $T(G)$ is an \emph{ordered cut-tree} if children of each cut vertex of $T(G)$ are ordered in increasing order of their size from left to right. Let $T(G)$ be an ordered cut-tree of $G$. We denote the parent of a vertex $v\in V(T)$ by $\textsc{parent}(v)$. For a block vertex $B\in\mathcal{B}$, $V(B)$ denotes the set of vertices in the block $B$ in $G$. Let $v\in\mathcal{C}$ be a cut vertex of $T(G)$ with block children $B_1, B_2, \ldots, B_r$. The set of blocks containing a vertex $v$ is denoted by $\mathcal{B}(v)$, i.e., $\mathcal{B}(v)=\{\textsc{parent}(v), B_1, B_2, \ldots, B_r\}$. The set of all the vertices in $\mathcal{B}(v)$ is denoted by $V(\mathcal{B}(v))$. Let $f:S\to[k]$ be a partial coloring of $G$, where $S\subseteq V(G)$. For a vertex subset $S'\subseteq V(G)$, we define $\textsc{ColorSet}(S')=\{f(v)\mid v\in S\cap S'\}$, that is, $\textsc{ColorSet}(S')$ denotes the set of colors assigned by $f$ to the colored vertices in $S'$. 

\begin{obs}\label{non-complete}
    If $G$ is a complete graph, then $\chi_{\text{so}}=\omega(G)$.
\end{obs}

From now on, we assume that $G$ is not a complete graph. We present an algorithm that computes an optimal strong odd coloring of $G$.

Let $G$ be a non-complete block graph and let $T(G)$ be its cut-tree rooted at a cut vertex $v$. For every cut vertex $x$ of $G$, let $B_1^x,B_2^x,\ldots,B_{r_x}^x$ be its block children in $T(G)$ with $|B_1^x|\le |B_2^x|\le\cdots\le |B_{r_x}^x|$. We describe the algorithm by explaining its two main procedures, $\textsc{ColorRoot}$ and $\textsc{ColorCutVertex}$.  First, Algorithm~\ref{block_algo} calls the procedure $\textsc{ColorRoot}$. 
\begin{enumerate}
	\item $\textsc{ColorRoot}$ assigns the color $1$ to $v$.
	\item Since $B_{r_v}^v$ is a clique and $v$ is already colored, any strong odd coloring needs at least $|B_{r_v}^v|-1$ colors to color all the uncolored vertices of $B_{r_v}^v$. $\textsc{ColorRoot}$ constructs a set $\mathcal{T}_v$ such that $\mathcal{T}_v$ has $|B_{r_v}^v|-1$ colors and each color $\alpha\in\mathcal{T}_v$ is distinct from $f(v)$.
	\item $\textsc{ColorRoot}$ colors the vertices of blocks $B_{1}^v, B_{2}^v,\ldots, B_{r_v-2}^v$ and some vertices of $B_{r_v-1}^v$ so that any color $\alpha\in\mathcal{T}_v$ that appears on the vertices of $B_{1}^v, B_{2}^v,\ldots, B_{r_v-1}^v$ appears an even number of times. Later, all these colors are used on the vertices of $B_{r_v}^v$. Thus, the parity of each color $\alpha\in\mathcal{T}_v$ becomes odd in $N(v)$.
	\item $\textsc{ColorRoot}$ colors the remaining vertices of $B_{r_v-1}^v$ using new colors.
	\item $\textsc{ColorRoot}$ colors $B_{r_v}^v$ with colors from the set $\mathcal{T}_v$.
\end{enumerate}
Therefore, $\textsc{ColorRoot}$ colors the root vertex $v$ of $T(G)$ and all the block children of $v$. When all the block children of $v$ are colored, $v$ satisfies the strong odd condition.

\begin{algorithm}[H]\label{block_algo}
	\relsize{-1}{
		\DontPrintSemicolon
		
		\KwIn{A non-complete block graph $G=(V,E)$.}
		\KwOut{A strong odd coloring $f$ of $G$.}
		
		Let $T(G)$ be an ordered cut-tree of $G$, rooted at an arbitrary cut vertex $v$;\
		
		For every cut vertex $x$, let $B_1^x,B_2^x,\ldots,B_{r_x}^x$ be its block children in $T(G)$ with $|B_1^x|\le |B_2^x|\le\cdots\le |B_{r_x}^x|$;\\
		Traverse $T(G)$ in breadth-first search order;\\
		Set $S=\emptyset$;\\
		
		$f_1, S'=\textsc{ColorRoot}(G, T(G))$;\\
		$f=f\cup f_1$;\\
		$S=S\cup S'$;\\
		
		\While{$($there exists a cut vertex $x$ such that its block children are uncolored$)$}{
			$f_1, S'=\textsc{ColorCutVertex}(G, T(G), S, f, x)$;\\
			$f=f\cup f_1$;\\
			$S=S\cup S'$;\\
	}}
	\Return{$f$}
	\caption{\textsc{StrongOdd-ColoringBlock($G$)}}
\end{algorithm}

\begin{procedure}[H]
	\relsize{-1}{ 
		\DontPrintSemicolon
		\KwIn{A block graph $G$, its ordered cut-tree $T(G)$.}
		\KwOut{A partial coloring $f_1$ of $G$, a set of colors $S$.}
		
		Set $S=\{1, 2, \ldots, |V(B_{r_v}^v)|\}$;\\
		Set $f_1(v)=1$;\\
		
		Let $\mathcal{T}_v\subseteq S\setminus\{f_1(v)\}$ be such that $|\mathcal{T}_v|=|V(B_{r_v}^v)|-1$.\\
		\For{$(i\gets1$ to $r_v-2)$}{
			\While{$($there is an uncolored vertex $z\in V(B_i^v))$}{
				Assign a color $c\in\mathcal{T}_v\setminus\textsc{ColorSet}(V(B_i^v))$ to $z$ and an uncolored vertex in $B_{i+1}$.
			}  
		}
		Let $\mathcal{U}$ be the set of uncolored vertices of $B_{r_v-1}^v$.\\
		\If{$(|\mathcal{U}|\ge |S\setminus(\mathcal{T}_v\cup f_1(v))|)$}{
			Add $|\mathcal{U}|-|S\setminus(\mathcal{T}_v\cup f_1(v))|$ number of new colors to $S$.
		}
		\While{$($there is an uncolored vertex $z\in B_{r_v-1}^v)$}{
			Assign a color $c\in S\setminus\textsc{ColorSet}(V(\mathcal{B}(v)))$ to $z$.
		}
		
		\While{$($there is an uncolored vertex $z\in B_{r_v}^v)$}{
			Assign a color $c\in \mathcal{T}_v\setminus\textsc{ColorSet}(V(B_{r_v}^v))$ to $z$.\\
		}

		\Return{$f_1, S$}
		\caption{()\textsc{ColorRoot}$(G, T(G))$}}
\end{procedure}

\begin{claim} \label{root_odd_condition}
	The vertex $v$ satisfies the strong odd condition under the partial coloring $f_1$ returned by the procedure $\textsc{ColorRoot}$.
\end{claim}
\begin{proof}
	Note that $\textsc{ColorRoot}$ selects a set $\mathcal{T}_v$ containing $|V(B_{r_v}^v)|-1$ colors. Then $\textsc{ColorRoot}$ colors the blocks $B_{1}^v, B_{2}^v,\ldots, B_{r_v-2}^v$ and some vertices of $B_{r_v-1}^v$ with colors from $\mathcal{T}_v$ until each block $B\in\{B_1^v, B_{2}^v, \ldots, B_{r_v-2}^v\}$ is colored. Since $|V(B_1^v)|\le|V(B_2^v)|\le\ldots\le|V(B_{r_v}^v)|$, there is always an uncolored vertex in $B_{i+1}^v$ whenever there is an uncolored vertex in $B_i^v$.  Whenever the procedure $\textsc{ColorRoot}$ assigns a color $c$ on an uncolored vertex of some block $B_i^v\in\{B_1^v,B_{2}^v, \ldots, B_{r_v-2}^v\}$, it colors an uncolored vertex in $B_{i+1}^v$ with color $c$. Consequently, each color $c\in\mathcal{T}_v$ appears an even number of times in $N(v)\setminus V(B_{r_v}^v)$. Note that $\textsc{ColorRoot}$  colors the uncolored vertices of $B_{r_v-1}^v$ with colors distinct from $f_1(v)$ and colors in  $\mathcal{T}_v$, and these colors are not used on $V(B_{r_v}^v)$. Finally, $\textsc{ColorRoot}$ colors the block $B_{r_v}^v$ with colors from the set $\mathcal{T}_v$. Since $|\mathcal{T}_v|=|V(B_{r_v}^v)|-1$ and $v\in V(B_{r_v}^v)$ is already colored, $\mathcal{T}_v$ has sufficient colors to color each vertex of $B_{r_v}^v$. Consequently, each color of the set $\mathcal{T}_v$ appears an odd number of times in $N(v)$ and any color $c\notin\mathcal{T}_v$ appearing on $B_{r_v-1}^v$ appears exactly once in $N(v)$. Hence, $v$ satisfies the strong odd condition under the partial coloring $f_1$ returned by the procedure $\textsc{ColorRoot}$.
\end{proof}

Then, Algorithm~\ref{block_algo} traverses the tree $T(G)$ in breadth-first search order. Whenever Algorithm~\ref{block_algo} encounters a cut vertex $x$ such that its block children are uncolored, it calls the procedure $\textsc{ColorCutVertex}$. 
\begin{enumerate}
	\item Note that $\textsc{parent}(x)$ is already colored. Therefore, $x$ is already colored.
	\item Since $B_{r_x}^x$ is a clique and $x$ is already colored, any strong odd coloring needs at least $|B_{r_x}^x|-1$ colors to color all the uncolored vertices of $B_{r_x}^x$. $\textsc{ColorCutVertex}$ constructs two disjoint sets $\mathcal{T}$ and $\mathcal{M}$ such that $|\mathcal{T}\cup\mathcal{M}|=|B_{r_x}^x|-1$.
	\item If $|\textsc{ColorSet}(\textsc{parent}(x))\setminus\{f(x)\}|\geq|B_{r_x}^x|-1$, then $\textsc{ColorCutVertex}$ selects $|B_{r_x}^x|-1$ colors from $\textsc{ColorSet}(\textsc{parent}(x))\setminus\{f(x)\}$ and takes $\mathcal{M}=\emptyset$.
	\item If $|\textsc{ColorSet}(\textsc{parent}(x))\setminus\{f(x)\}|<|B_{r_x}^x|-1$, then $\textsc{ColorCutVertex}$ selects all the colors from $\textsc{ColorSet}(\textsc{parent}(x))\setminus\{f(x)\}$ in the set $\mathcal{T}$ and selects $|V(B_{r_x}^x)|-|\mathcal{T}|-1$ colors distinct from colors of the set $\mathcal{T}$ in the set $\mathcal{M}$.
	\item $\textsc{ColorCutVertex}$ tries to make the parity of as many colors of the set $\mathcal{T}$ even in $N(x)\setminus V(B_{r_x}^x)$. Once parity of each color is even in $N(x)\setminus V(B_{r_x}^x)$, $\textsc{ColorCutVertex}$ colors remaining vertices of $N(x)\setminus V(B_{r_x}^x)$ with colors from $\mathcal{T}\cup\mathcal{M}$ keeping the parity even in $N(x)\setminus V(B_{r_x}^x)$. Later, all these colors are used on the vertices of $B_{r_x}^x$. Thus, the parity of each color of the set $\mathcal{T}\cup\mathcal{M}$ becomes odd in $N(x)$.
	\item $\textsc{ColorCutVertex}$ colors the vertices of blocks $B_{1}^x, B_{2}^x,\ldots, B_{r_x-1}^x$ one by one using the colors from the set $\mathcal{T}$ exactly once, ensuring that the parity of the colors used becomes even. 
	\item If each color of the set $\textsc{ColorSet}(\textsc{parent}(x)\setminus f(x))$ is used on a vertex of $B_{1}^x, B_{2}^x,\ldots, B_{r_x-1}^x$ and all blocks are colored except $B_{r_x-1}^x$ and $B_{r_x}^x$, then $\textsc{ColorCutVertex}$ uses new colors distinct from $\textsc{ColorSet}(V(\mathcal{B}_v))$ to color the uncolored vertices of $B_{r_x-1}^x$. 
	
	\item If each color of the set $\textsc{ColorSet}(\textsc{parent}(x)\setminus f(x))$ is used on a vertex of $B_{1}^x, B_{2}^x,\ldots, B_{r_x-2}^x$ and all blocks are not colored yet, then $\textsc{ColorCutVertex}$ colors the uncolored vertices of $B_{1}^x, B_{2}^x,\ldots, B_{r_x-2}^x$ and some vertices of $B_{r_x-1}^x$ using the colors of the set $\mathcal{T}\cup\mathcal{M}$ ensuring that each color of the $\mathcal{T}\cup\mathcal{M}$ appears an even number of times in $N(x)\setminus V(B_{r_x}^x)$. 
	
	\item If each block $B\in\{B_{1}^x, B_{2}^x,\ldots, B_{r_x-1}^x\}$ is colored and there are some colors in $\mathcal{T}$ which are not used on the vertices of $B_{1}^x, B_{2}^x,\ldots, B_{r_x-1}^x$, then $\textsc{ColorCutVertex}$ first colors the vertices of $B_{r_x}^x$ with the colors that apears an even number of times in $N(x)\setminus V( B_{r_x}^x)$. Note that the set of colors appearing an even number of times in $N(x)\setminus V( B_{r_x}^x)$ is $\mathcal{T}'\cup\mathcal{M}$, where $\mathcal{T'}$ is the set of colors of $\mathcal{T}$ that appears on the vertices on $B_{1}^x, B_{2}^x,\ldots, B_{r_x-1}^x$. If there are some uncolored vertices in $B_{r_x}^x$, $\textsc{ColorCutVertex}$ colors the uncolored vertices of $B_{r_x}^x$ with new colors distinct from $\textsc{ColorSet}(V(\mathcal{B}_v))$.
	
	\item After all these steps, $\textsc{ColorCutVertex}$ colors the remaining uncolored vertices of $B_{r_x-1}^x$ by using new colors.
	
	\item Finally, $\textsc{ColorCutVertex}$ colors the uncolored vertices of $B_{r_x}^x$ by using colors from the set $\mathcal{T}\cup\mathcal{M}$.
\end{enumerate}
Thus, $\textsc{ColorCutVertex}$ colors all block children of $x$. When all the block children of $x$ are colored, $x$ satisfies the strong odd condition.

\begin{procedure}[H]
	\relsize{-1}{ 
		\DontPrintSemicolon
		\KwIn{A block graph $G$, its ordered cut-tree $T(G)$, a set of colors $S$, a partial coloring $f_1$ such that all the vertices that come before $x$ in breadth-first search order of $T(G)$ are colored, and a cut vertex $x$.}
		\KwOut{A partial coloring $f_1$ of $G$ such that $N(x)$ is colored under $f_1$ and a set of colors $S$.}
		
		Let $\mathcal{R}=\textsc{ColorSet}(V(\textsc{parent}(x)))\setminus\{f_1(x)\}$.\\
		\If{$(|\mathcal{R}|\geq|V(B_{r_x}^x)|-1)$}{
			Let $\mathcal{T}\subseteq\mathcal{R}$ be such that $|\mathcal{T}|=|V(B_{r_x}^x)|-1$.\\
			Let $\mathcal{M}=\emptyset$.\\
		}\Else{
			Let $\mathcal{T}=\mathcal{R}$.\\
			\If{$(|V(B_{r_x}^x)|\ge|S|)$}{Add $|V(B_{r_x}^x)|-|S|$ number of new colors to $S$.\\}
			Let $\mathcal{M}\subseteq S\setminus\mathcal{R}$ be such that $|\mathcal{M}|=|V(B_{r_x}^x)|-|\mathcal{T}|-1$.\\
		}
		
		Let $\mathcal{T'}=\mathcal{T}$.\\

		\For{$(i \gets 1$ to $r_x-2)$}{
			\While{$($there is an uncolored vertex $z\in V(B_i^x))$}{
				
				\If{$(\mathcal{T'}\ne\emptyset)$}{
					Assign a color $c\in\mathcal{T'}$ to $z$.\\
					$\mathcal{T'}=\mathcal{T'}\setminus\{c\}$.\\
				}\Else{
					Assign a color $c\in(\mathcal{T}\cup\mathcal{M})\setminus\textsc{ColorSet}(V(B_i^x))$ to $z$ and an uncolored vertex in $B_{i+1}^x$.\\
				}
			}  
		}
		\If{$((|V(B_{r_x-1}^x)|-1>|\mathcal{T}'|)\text{ and }(S\setminus\textsc{ColorSet}(V(\mathcal{B}(x)))=\emptyset)$}{Add $|V(B_{r_x-1}^x)|-|\mathcal{T}|-1$ new colors to $S$.\\}
		\While{$($there is an uncolored vertex $z\in B_{r_x-1}^x)$}{
			\If{$(\mathcal{T}'\ne\emptyset)$}{
				Assign a color $c\in\mathcal{T'}$ to $z$.\\
				$\mathcal{T'}=\mathcal{T'}\setminus\{c\}$.\\
			}\Else{
				Assign color $c\in S\setminus\textsc{ColorSet}(V(\mathcal{B}(x)))$ to $z$.
			} 
		}
		Let $\mathcal{C}$ be the set of colors that appear an even number of times in colored vertices of $N(x)$.\\
		\If{$((|V(B_{r_x}^x)|-1>|\mathcal{C}|)\text{ and }(S\setminus\textsc{ColorSet}(V(\mathcal{B}(x)))=\emptyset))$}{Add $|V(B_{r_x}^x)|-|\mathcal{C}|-1$ new colors to $S$.\\}
		\While{$($there is an uncolored vertex $z\in B_{r_x}^x)$}{
			\If{$(\mathcal{C}\setminus\textsc{ColorSet}(V(B_{r_x}^x))\ne\emptyset)$}{
				Assign color $c\in \mathcal{C}\setminus\textsc{ColorSet}(V(B_{r_x}^x))$ to $z$.\\
			}\Else{
				Assign color $c\in S\setminus\textsc{ColorSet}(V(\mathcal{B}(x)))$ to $z$.
			}
		}
		
		\Return{$f_1, S$}
		\caption{()\textsc{ColorCutVertex}$(G, T(G), S, f_1, x)$}}
\end{procedure}

\begin{figure}[h]
	\begin{center}
		\includegraphics[width=\linewidth]{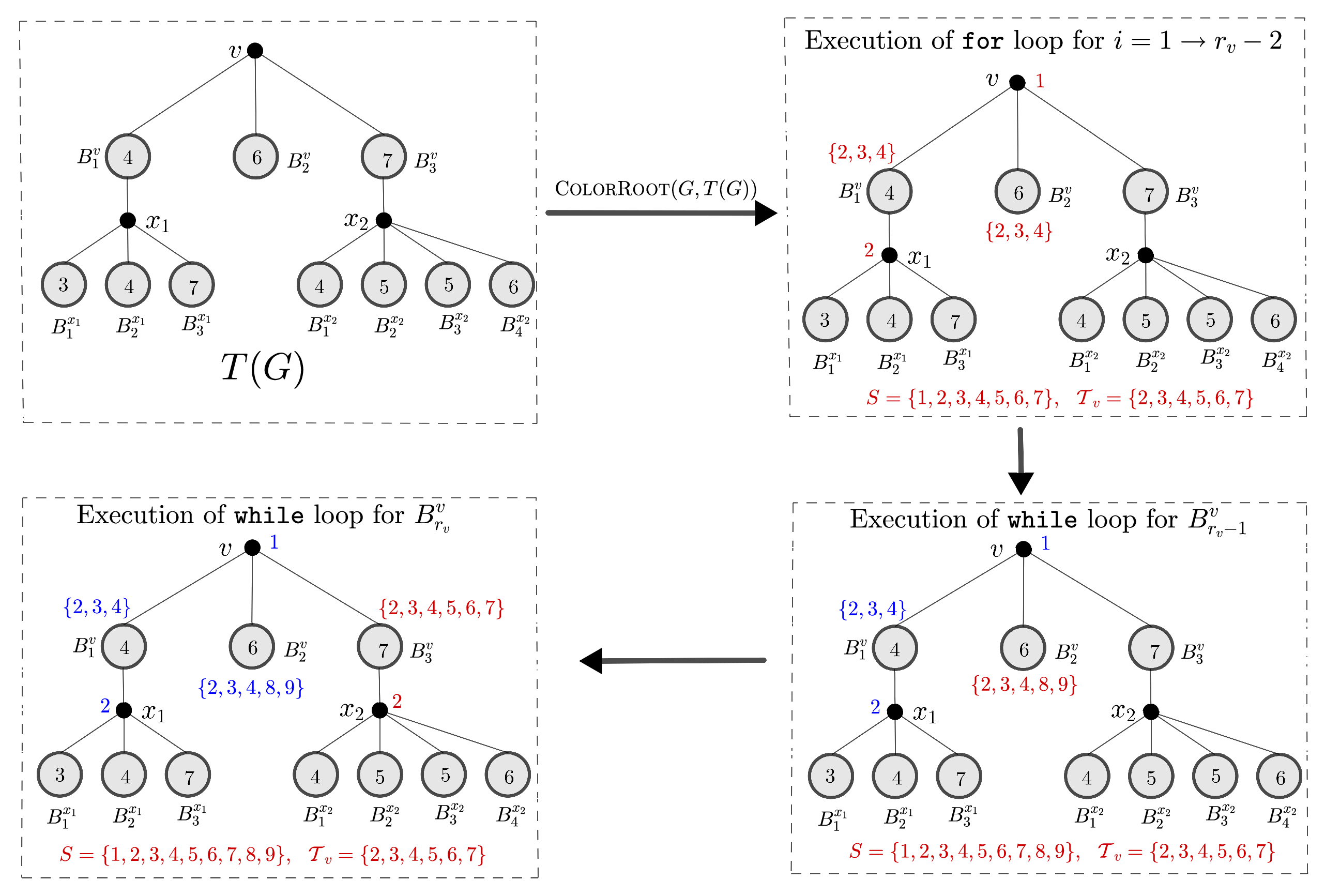}
		\caption{Execution of $\textsc{ColorRoot}$ on a cut-tree. Each circle represents a block, and the number inside the circle denotes the number of vertices in the corresponding block, including its cut-vertices. The list associated with a block specifies the colors assigned to its uncolored vertices (the parent of each block $B$ is colored before $B$ is processed). Lists shown in blue indicate that each vertex of the corresponding block was colored in the previous step.}
		
		\label{fig_color_root}
	\end{center}	
\end{figure}

\begin{figure}[h]
	\begin{center}
		\includegraphics[width=\linewidth]{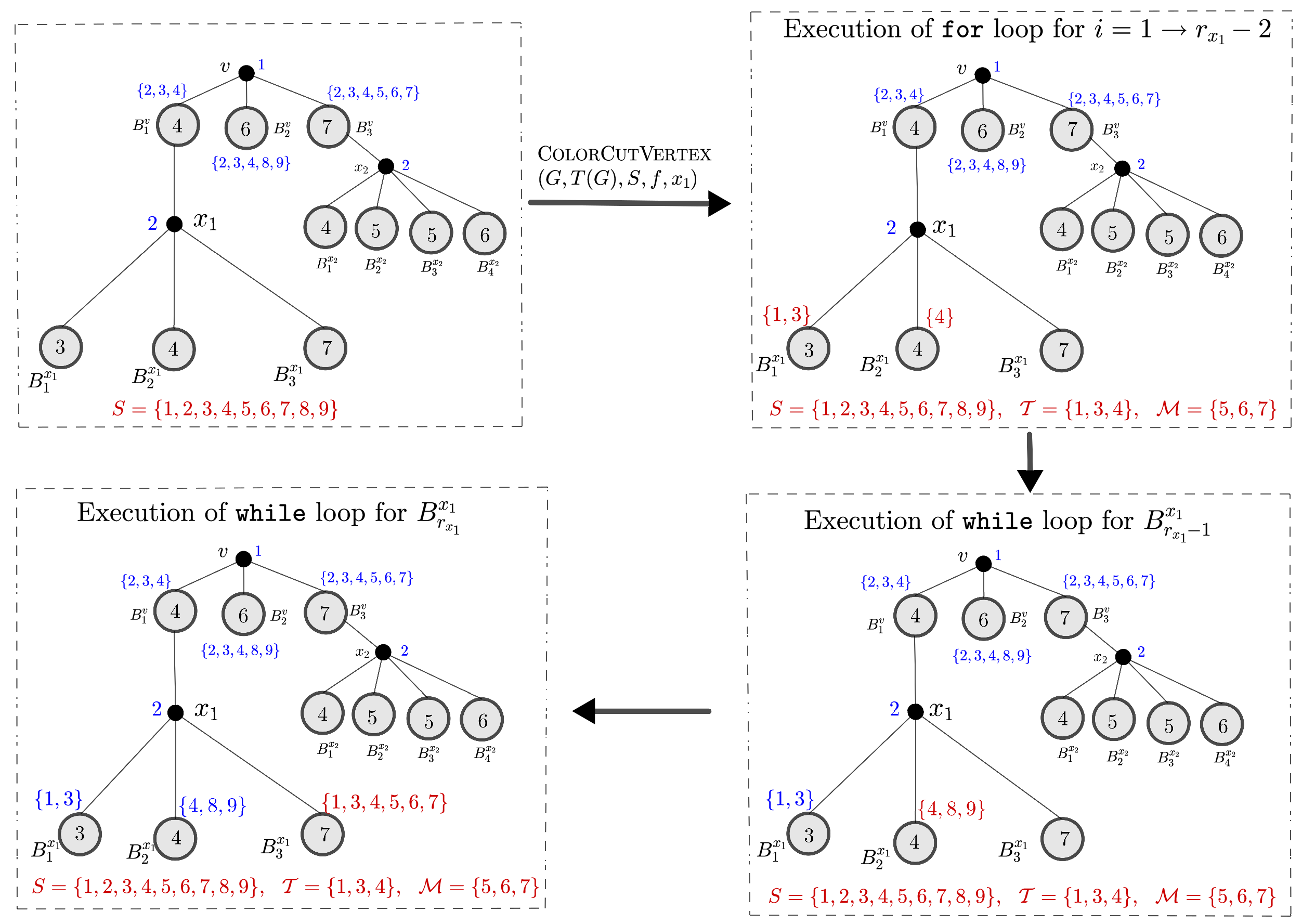}
		\caption{Execution of $\textsc{ColorCutVertex}$ for $x_1$. Each circle represents a block, and the number inside the circle denotes the number of vertices in the corresponding block, including its cut-vertices. The list associated with a block specifies the colors assigned to its uncolored vertices (the parent of each block $B$ is colored before $B$ is processed). Lists shown in blue indicate that each vertex of the corresponding block was colored in the previous step or in the partial strong odd coloring given as input to the procedure.}
		\label{fig_color_cut}
	\end{center}	
\end{figure}

\begin{figure}[h]
	\begin{center}
		\includegraphics[width=\linewidth]{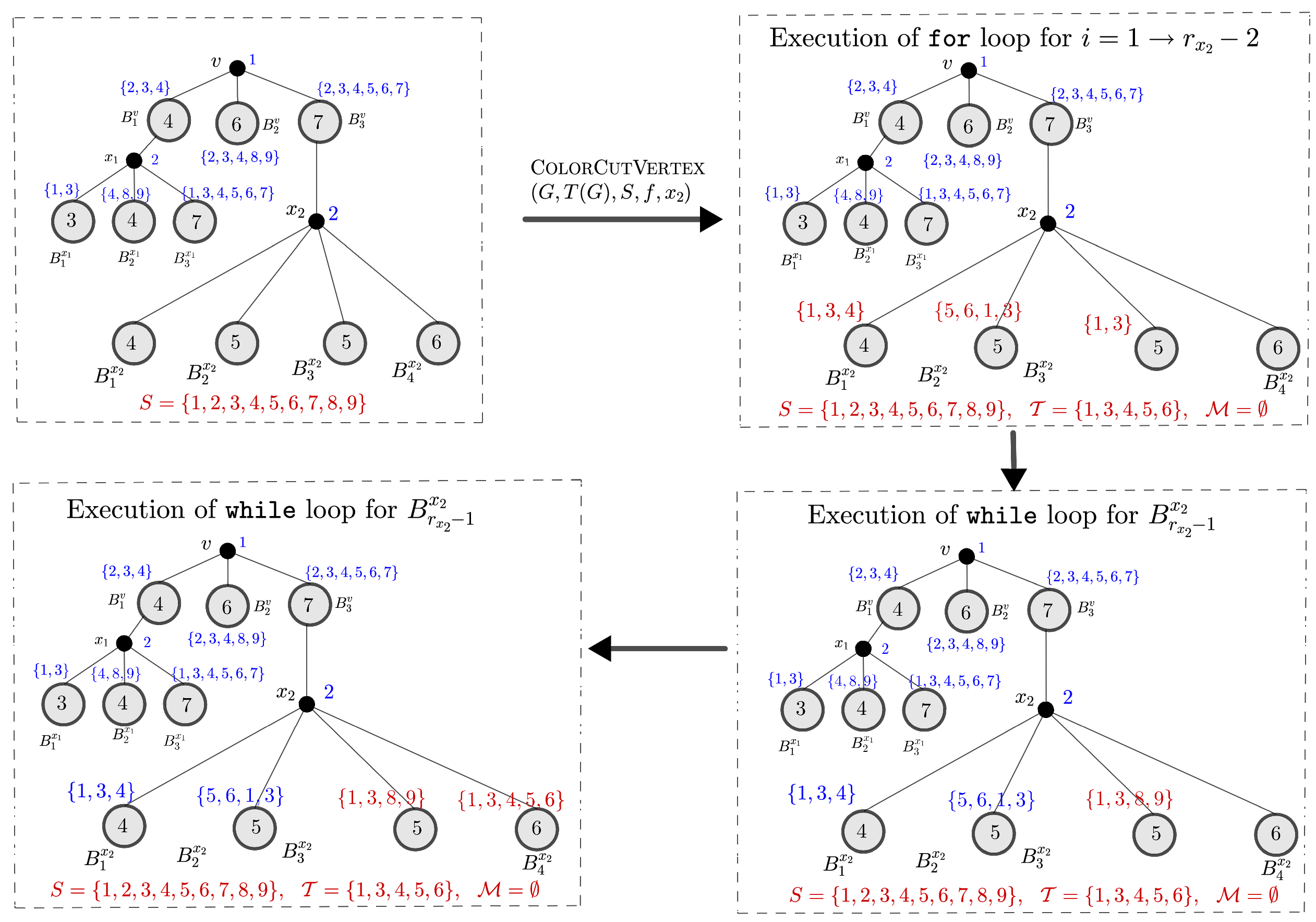}
		\caption{Execution of $\textsc{ColorCutVertex}$ for $x_2$. Each circle represents a block, and the number inside the circle denotes the number of vertices in the corresponding block, including its cut-vertices. The list associated with a block specifies the colors assigned to its uncolored vertices (the parent of each block $B$ is colored before $B$ is processed). Lists shown in blue indicate that each vertex of the corresponding block was colored in the previous step or in the partial strong odd coloring given as input to the procedure.}
		\label{fig_color_cut2}
	\end{center}	
\end{figure}

\begin{claim}\label{cut_vertex_odd_condition}
	Let $x\neq v$ be a cut vertex. The vertex $x$ satisfies the strong odd condition under the partial coloring $f_1$ returned by the procedure $\textsc{ColorCutVertex}$.
\end{claim}
\begin{proof}
	Note that $\textsc{parent}(x)$ is already colored. Since $B_{r_x}^x$ is a clique and $x$ is already colored, any strong odd coloring needs at least $|B_{r_x}^x|-1$ colors to color all the uncolored vertices of $B_{r_x}^x$. $\textsc{ColorCutVertex}$ constructs two disjoint sets $\mathcal{T}$ and $\mathcal{M}$ such that $|\mathcal{T}\cup\mathcal{M}|=|B_{r_x}^x|-1$. If $|\textsc{ColorSet}(\textsc{parent}(x))\setminus\{f(x)\}|\geq|B_{r_x}^x|-1$, then $\textsc{ColorCutVertex}$ selects $|B_{r_x}^x|-1$ colors from $\textsc{ColorSet}(\textsc{parent}(x))\setminus\{f(x)\}$ in $\mathcal{T}$ and takes $\mathcal{M}=\emptyset$. If $|\textsc{ColorSet}(\textsc{parent}(x))\setminus\{f_1(x)\}|<|B_{r_x}^x|-1$, then $\textsc{ColorCutVertex}$ takes all the colors from $\textsc{ColorSet}(\textsc{parent}(x))\setminus\{f(x)\}$ in $\mathcal{T}$ and selects $|V(B_{r_x}^x)|-|\mathcal{T}|-1$ colors in the set $\mathcal{M}$. $\textsc{ColorCutVertex}$ first colors $|\mathcal{T}|$ number of vertices of the blocks $B_1^x,B_{2}^x,\ldots,B_{r_x-1}^x$ in breadth first search order. 
	
	\smallskip
	\noindent \textbf{Case $1$:} When all colors of $\mathcal{T}$ are exhausted, i.e., $\mathcal{T}'=\emptyset$ at some iteration. 
	
	\smallskip
	
	Note that the parity of each color $c\in\mathcal{T}\cup\mathcal{M}$ becomes even in $N(x)\setminus V(B_{r_x}^x)$. Since $\textsc{ColorCutVertex}$ is processing blocks in breadth first search order and $|V(B_1^x)|\le|V(B_2^x)|\le\ldots\le|V(B_{r_x}^x)|$, $B_{i+1}^x$ is uncolored whenever $B_i^x$ is uncolored. Moreover, the number of uncolored vertices in $B_{i+1}^x$ is at least the number of uncolored vertices in $B_{i}^x$. Let $v\in V(B_i^x)$ is uncolored, for some $i\in[r_{x}-2]$. $\textsc{ColorCutVertex}$ assigns a color $c\in\mathcal{T}\cup\mathcal{M}$ to $v$ and an uncolored vertex in $B_{i+1}^x$. Since $|\mathcal{T}\cup\mathcal{M}|=|V(B_{r_x}^x|-1$, each block $B_i^x\in\{B_1^x, B_2^x, \ldots, B_{r_x-2}^x\}$ can be colored with colors from $\mathcal{T}\cup\mathcal{M}$ maintaining properness. After coloring each block $B_i^x\in\{B_1^x, B_2^x, \ldots, B_{r_x-2}^x\}$, $\textsc{ColorCutVertex}$ colors uncolored vertices of $B_{r_x-1}^x$ with new colors which are distinct from colors in $\mathcal{T}\cup\mathcal{M}\cup\{f(x)\}$. Since $|\mathcal{T}\cup\mathcal{M}|=|V(B_{r_x}^x)|-1$ and $x\in V(B_{r_x}^x)$ is already colored, $\mathcal{T}\cup\mathcal{M}$ has sufficient colors to color each vertex of $B_{r_x}^x$. $\textsc{ColorCutVertex}$ colors vertices of $B_{r_x}^x$ with colors from $\mathcal{T}\cup\mathcal{M}$. Therefore, each color $c\in\mathcal{T}\cup\mathcal{M}$ appears an even number of times in $N(x)$, and each color $c\notin\mathcal{T}\cup\mathcal{M}$ appearing on $N(x)$ appears exactly once in $N(x)$. Thus, $x$ satisfies the strong odd condition under the partial coloring $f_1$ returned by the procedure $\textsc{ColorCutVertex}$.
	
	\smallskip
	\noindent \textbf{Case $2$:} When all colors of $\mathcal{T}$ are not exhausted, i.e., $\mathcal{T}'\ne\emptyset$ at any iteration.
	
	\smallskip
	
	Let $\mathcal{T'}\subset\mathcal{T}$ be the set of colors that are used on the vertices of $B_{1}^x, B_2^x, \ldots, B_{r_x-1}^x$. Note that each vertex of the blocks $B_{1}^x, B_2^x, \ldots, B_{r_x-1}^x$ is colored. Since no color of the set $\mathcal{M}$ has appeared in $N(x)\setminus V(B_{r_x}^x)$, the set of colors appearing an even number of times in $N(x)\setminus V(B_{r_x}^x)$ is $\mathcal{C}=\mathcal{T'}\cup\mathcal{M}$.
	$\textsc{ColorCutVertex}$ colors the vertices of $V(B_{r_x}^x)$ with colors from the set $\mathcal{C}$ and then uses new colors distinct from the colors used in $N(x)\setminus V(\mathcal{B}(x))$, maintaining properness. Since $|\mathcal{T}\cup\mathcal{M}|=|V(B_{r_x}^x)|-1$, we have $|\mathcal{T}'\cup\mathcal{M}|\le|V(B_{r_x}^x)|-1$.Therefore, each color $c\in\mathcal{C}$ appears on some vertex of $B_{r_x}^x$. Consequently, parity of each color $c\in\mathcal{C}$ becomes odd in $N(x)$. Note that each color $c\notin\mathcal{C}$ appearing on $N(x)$ appears exactly once on $N(x)$. Thus, $x$ satisfies the strong odd condition under the partial coloring $f_1$ returned by the procedure $\textsc{ColorCutVertex}$.
\end{proof}

Since $N(v)$ is a clique for each $v\in V(G)\setminus\mathcal{C}$, any color that appears in $N(v)$ appears uniquely under any proper coloring. Consequently, we have the following observation.

\begin{obs}\label{block-vertices_odd_condition}
	Let $\varphi$ be a proper $k$-coloring of $G$. Then each vertex $v\in V(G)\setminus\mathcal{C}$ satisfies the strong odd condition under $\varphi$.
\end{obs}

\begin{lemma}\label{so_coloring}
	Algorithm~\ref{block_algo} returns a strong odd coloring of $G$.
\end{lemma}
\begin{proof}
	We first show that $f$ is a proper coloring. Note that Algorithm~\ref{block_algo} calls $\textsc{ColorRoot}$ to color the block children of root $v$. Algorithm~\ref{block_algo} traverses $T(G)$ in breadth-first search order and calls $\textsc{ColorCutVertex}$ to color the block children to any cut vertex $x\neq v$ whose block children are not already colored. Whenever $\textsc{ColorRoot}$ or $\textsc{ColorCutVertex}$ colors an uncolored vertex $z$ of a block $B$, $z$ is colored with a color that does not appear in $\textsc{ColorSet}(V(B))$. Thus, no two vertices of the same block receive the same color. Therefore, $f(u)\neq f(v)$ whenever $uv\in E(G)$. Hence, $f$ is a proper coloring of $G$. 
	
	It remains to show that each vertex $v\in V(G)$ satisfies the strong odd condition. Let $u\in V(G)\setminus\mathcal C$. By Observation~\ref{block-vertices_odd_condition}, $u$ satisfies the strong odd condition. Now, let $u\in\mathcal C$. By Claims~\ref{root_odd_condition} and \ref{cut_vertex_odd_condition}, $x$ satisfies the strong odd condition under the partial colorings returned by $\textsc{ColorRoot}$ and $\textsc{ColorCutVertex}$. Since no vertex $z\in N(x)$ is recolored by Algorithm~\ref{block_algo}, $x$ satisfies the strong odd condition under the coloring $f$ returned by Algorithm~\ref{block_algo}.
	
	Therefore, $f$ is a strong odd coloring of $G$.\qed
\end{proof}

\begin{lemma}\label{block:optimal}
	Let Algorithm~\ref{block_algo} return a coloring $f$ of $G$ using $k$ colors. Then $\chi_{\text{so}}(G)= k$.
\end{lemma}
\begin{proof}
	Note that there exists a cut vertex $x\in V(T)$ such that Algorithm~\ref{block_algo} enlarges $S$ to contain $k$  colors while processing the blocks containing $x$. Let $B_1^x,B_2^x,\ldots,B_{r_x}^x$ be the block children of $x$ such that $|B_1^x|\leq |B_2^x|\leq\cdots\leq |B_{r_x}^x|$. Let $|V(B_{r_x}^x)|=q$. 
	
	\smallskip
	\noindent \textbf{Case $1$:} $x=v$, i.e, $x$ is the root of $T(G)$. 
	
	\smallskip
	
	Let $\varphi$ be any strong odd coloring of $G$. Without loss of generality, assume that $\varphi(x)=f(x)=1$ and colors $2, 3, \ldots, q$ are assigned to the vertices of $V(B_{r_x}^x)\setminus\{x\}$ under $\varphi$ and $f$. Let $\mathcal{T}=\{2,3,\ldots,q\}$. Note that each color $c\in\mathcal{T}$ appears an even number of times in $N(x)\setminus V(B_{r_x}^x)$ and each color $c\notin\mathcal{T}$ that appears in $N(x)$ occurs an odd number of times in $N(x)\setminus V(B_{r_x}^x)$. Let $\alpha_1, \alpha_2, \ldots, \alpha_m$ be the colors that appear an odd number of times in $N(x)\setminus V(B_{r_x}^x)$ under $f$. It is easy to see that $k=q+m$. Since each color can appear at most once in a block in a strong odd coloring, each color $c\in\mathcal{T}$ appears in an even number of blocks in $\{B_1^x,B_2^x,\ldots,B_{r_x-1}^x\}$ and each color $c\notin\mathcal{T}$ that appears in $N(x)$ occurs an odd number of blocks in $\{B_1^x,B_2^x,\ldots,B_{r_x-1}^x\}$. Therefore, $\varphi$ must use at least $m$ number of new colors to color the blocks $\{B_1^x,B_2^x,\ldots,B_{r_x}^x\}$ keeping the parity constraint.
	
	\smallskip
	\noindent \textbf{Case $2$:} $x\ne v$, i.e, $x$ is not the root of $T(G)$.
	
	\smallskip
	
	Let $\varphi$ be any strong odd coloring of $G$, and let $|V(\textsc{parent}(x))|=q$. Without loss of generality, assume that $\varphi(x)=f(x)=1$ and colors $2, 3, \ldots, q$ are assigned to the vertices of $V(\textsc{parent}(x))\setminus\{x\}$ under $\varphi$ and $f$. Let $\mathcal{T}=\{2,3,\ldots,q\}$. Note that each color $c\in\mathcal{T}$ appears an even number of times in $N(x)\setminus V(\textsc{parent}(x))$ and each color $c\notin\mathcal{T}$ that appears in $N(x)$ occurs an odd number of times in $N(x)\setminus V(\textsc{parent}(x))$. Let $\alpha_1, \alpha_2, \ldots, \alpha_m$ be the colors that appear an odd number of times in $N(x)\setminus V(\textsc{parent}(x))$ under $f$. It is easy to see that $k=q+m$. Since each color can appear at most once in a block in a strong odd coloring, each color $c\in\mathcal{T}$ appears in an even number of blocks in $\{B_1^x,B_2^x,\ldots,B_{r_x}^x\}$ and each color $c\notin\mathcal{T}$ that appears in $N(x)$ occurs an odd number of blocks in $\{B_1^x,B_2^x,\ldots,B_{r_x}^x\}$. Note that $B_{r_x}^x$ can use at most $|V(B_{r_x}^x)|-1$ colors from $\mathcal{T}$. If a color $c\in\mathcal{T}$ appears on $B_{r_x}^x$, then it must appear on an odd number of blocks in $\{B_1^x,B_2^x,\ldots,B_{r_x-1}^x\}$. If a color $c\notin\mathcal{T}$ appears on $B_{r_x}^x$, then it must appear in an even number of blocks in $\{B_1^x,B_2^x,\ldots,B_{r_x-1}^x\}$. Thus, $\varphi$ needs at least $m+q-1$ number of colors in $N(x)$ to maintain the parity constraints. 
	
	\medskip
	Therefore, if Algorithm~\ref{block_algo} returns a strong odd coloring of $G$ using $k$ colors, then $\chi_{\text{so}}(G)\ge k$. Since $f$ is a strong odd coloring by Lemma~\ref{so_coloring}, we have $\chi_{\text{so}}(G)\ge k$. Hence $\chi_{\text{so}}(G)= k$.
\end{proof}

\begin{theorem}\label{block_main_thm}$(\ast)$
Algorithm~\ref{block_algo} returns an optimal strong odd coloring of $G$ in $O(|V(G)|+|E(G)|)$ time.
\end{theorem}
\begin{proof}
Algorithm~\ref{block_algo} returns an optimal strong odd coloring by Lemmas~\ref{so_coloring} and~\ref{block:optimal} for every block graph $G$. 

Now, we show that Algorithm~\ref{block_algo} runs in $O(|V(G)|+|E(G)|)$ time. The cut-tree $T(G)$ can be constructed in $O(|V|+|E|)$ time using a depth-first search. The ordering of the block children of each cut vertex can be computed from their sizes in $O(|V|+|E|)$ time. The breadth-first traversal of $T(G)$ takes $O(|V|+|E|)$ time. Algorithm~\ref{block_algo} calls two procedures, \textsc{ColorRoot} and \textsc{ColorCutVertex}.

\smallskip
\noindent  \textbf{\textsc{ColorRoot}:}
\textsc{ColorRoot} colors the root $v$ and its block children. Note that \texttt{lines 1-3} and \texttt{lines 7-9} can be executed in $O(d(v))$ time. Each vertex is colored at most once. Moreover, each iteration of each loop in \textsc{ColorRoot} colors at least one vertex in $N(v)$. Therefore, the total number of iterations of all loops in \textsc{ColorRoot} is $O(d(v))$. For $i\in[r_v-2]\cup\{r_v\}$, the $\textsc{ColorSet}(V(B_i^v))$ computed in $\texttt{for}$ and $\texttt{while}$ loop for an uncolored vertex $z\in V(B_i^v)$ can be computed in $O(d(z))$ time. Similarly, the $\textsc{ColorSet}(V(\mathcal{B}(v)))$ computed in $\texttt{while}$ loop for an uncolored vertex $z\in V(B_{r_v-1}^v)$ can be computed in $d(v)$ time. 

\smallskip
\noindent\textbf{\textsc{ColorCutVertex}:}
Consider one execution of \textsc{ColorCutVertex} for a cut vertex $x$. Note that \texttt{lines 1-10}, \texttt{lines 18-19}, and \texttt{lines 27-28} can be executed in $O(d(x))$ time. The procedure colors only the block children of $x$. Each vertex is colored at most once. Moreover, each iteration of each loop in \textsc{ColorCutVertex} colors at least one vertex in $N(v)$. Therefore, the total number of iterations of all loops in \textsc{ColorCutVertex} is $O(d(x))$. For $i\in[r_x-2]\cup\{r_x\}$, the $\textsc{ColorSet}(V(B_i^x))$ computed in $\texttt{for}$ and $\texttt{while}$ loop for an uncolored vertex $z\in V(B_i^x)$ can be computed in $d(z)$ time. Similarly, the $\textsc{ColorSet}(V(\mathcal{B}(x)))$ computed in $\texttt{while}$ loop for an uncolored vertex $z\in V(B_{r_x-1}^v)$ can be computed in $d(x)$ time. 

\smallskip
Note that $\textsc{ColorSet}(V(B_i^x))$ is computed only when an uncolored vertex $z\in B_i^x$ is colored for some cut vertex $x$. $\textsc{ColorSet}(V(\mathcal{B}(x)))$ can be computed before executing the loop for $B_{r_x-1}^x$ and can be updated whenever a vertex $y\in B_{r_x-1}^x$ is colored in $O(1)$ time. Therefore, all executions of \textsc{ColorRoot} and \textsc{ColorCutVertex} runs in time $O\left(\sum_{y\in V(G)}d(y)\right)$. Since $\sum_{y\in V(G)}d(y)=2|E(G)|$, Algorithm~\ref{block_algo} runs in time $O(|V|+|E|)$. 
\end{proof}

\section{Parameterized complexity by structural parameters}

We study the parameterized complexity of \textsc{strong odd $k$-colorability} when parameterized by treewidth, feedback vertex set, and clique-width of the input graph.

\begin{theorem}
    \textsc{strong odd $k$-colorability} parameterized by treewidth and $k$ admits an FPT algorithm with running time $O(n^{O(1)}\cdot(k\cdot9^k)^{\texttt{tw}+1})$.
\end{theorem}
\begin{proof}
We assume that a nice tree decomposition $\mathcal{T}=(T,\{X_i\}_{i\in V(T)})$ of $G$ with $|X_i|\le \texttt{tw}+1$ for each $i\in V(T)$ is given. Thus, each node of $T$ is of one of the types: (i) leaf node, (ii) introduce node, (iii) forget node, and (iv) join node. Before going into the algorithm, we define an operation that will be used in this section. We define a binary operation $\star$ such that $0\star0=0$, $0\star\texttt{odd}=\texttt{odd}$, $0\star\texttt{even}=\texttt{even}$, $\texttt{odd}\star\texttt{odd}=\texttt{even}$, $\texttt{odd}\star\texttt{even}=\texttt{odd}$, and $\texttt{even}\star\texttt{even}=\texttt{even}$. Now, we proceed with the dynamic programming algorithm for treewidth. Given a bag $X_i$ corresponding to the vertex $i\in V(T)$, a state for the bag is a tuple $(i, \texttt{col},\texttt{par})$, where $i$ determines the bag; $\texttt{col}:X_i\to[k]$ is a vertex coloring of $X_i$; $\texttt{par}:X_i\times[k]\to\{0,\texttt{odd}, \texttt{even}\}$ is a parity function. The value $\texttt{par}(v,c)$ indicates the current parity status of the color $c$ in $N(v)$ among the edges processed so far. Precisely, $\texttt{par}(v,c)=\texttt{odd}$ if $c$ appears an odd number of times, $\texttt{par}(v,c)=\texttt{even}$ if $c$ appears a non-zero even number of times, and $\texttt{par}(v,c)=0$ if $c$ has not yet appeared in $N(v)$.

We denote the dynamic programming table corresponding to the node $i\in V(T)$ by $\mathcal{S}_i$ and the value corresponding to the state $(i,\texttt{col}, \texttt{par})$ is denoted by $\mathcal{S}_i(i,\texttt{col}, \texttt{par})$. For a node $i\in T$, let $V_i$ denote the set of vertices of $G$ appearing in bags of the subtree rooted at $i$. For a state $(i, \texttt{col}, \texttt{par})$, suppose there is a proper coloring $\varphi: V_i\to[k]$ such that (i) its restriction to $X_i$ is the coloring $\texttt{col}$, (ii) For each $(v,c)\in X_i\times[k]$, if $\texttt{par}(v, c)=\texttt{odd}$, $c$ appears an odd number of times in $N(v)\cap V_i$; if $\texttt{par}(v, c)=\texttt{even}$, $c$ appears an even($\ne0$) number of times $N(v)\cap V_i$; if $\texttt{par}(v, c)=0$, then $c$ does not appear at all in $N(v)\cap V_i$, (iii) $\texttt{par}(v,c)\in\{0,\texttt{odd}\}$ for each $c\in[k]$ if $i$ is a forget node forgetting the vertex $v$. Then $\mathcal{S}_i(i,\texttt{col},\texttt{par})=1$. Otherwise, $\mathcal{S}_i(i,\texttt{col},\texttt{par})=0$. Intuitively, we are assigning $0$ to only those states, for which $\texttt{col}$ cannot be extended to a strong odd coloring of the graph $G$. For a coloring to be strong odd, it should be proper, and each vertex should satisfy the strong odd condition. Therefore, we have condition (i) and (iii). Note that $\texttt{par}$ should correctly store the parity of each color in the neighborhood of each vertex $v\in X_i$ for a state $(i,\texttt{col},\texttt{par})$. Therefore, we have condition (ii).

Our dynamic programming algorithm will compute $\mathcal{S}_i$ for each state corresponding to each node $i\in V(T)$. For the root $X_r=\{z\}$, if there is a state $(r, \texttt{col}, \texttt{par})$ such that $\texttt{par}(z,c)\in\{0, \texttt{odd}\}$ for each $c\in[k]$ and $\mathcal{S}(r, \texttt{col}, \texttt{par})=1$, then $G$ has a strong odd $k$-coloring. Now, we describe our dynamic programming in cases according to the types of nodes of the tree decomposition.

\textbf{Leaf node:} Let $i$ be a leaf node with $X_i=\{z\}$. Then,
$$\mathcal{S}_i(i,\texttt{col},\texttt{par}) = \begin{cases} 1, \text{ if } \texttt{par}(z,c)=0 \text{ for each } c\in[k];\\
0, \text{ otherwise};\end{cases}$$

For the correctness, note that no neighbor of $z$ is colored yet in the graph processed so far, as a leaf node only contains $z$.

\smallskip

\textbf{Introduce node:} Let $i$ be an introduce node with child $j$ introducing a vertex $v$. Note that $X_i=X_j\cup\{v\}$. Consider a state $(i, \texttt{col}, \texttt{par})$. Let $\texttt{col}(v)=\alpha$. Since $\mathcal{T}$ is a tree decomposition, each vertex $u\in N(v)\cap V_i$ is in the bag $X_i$. Thus, $\texttt{par}(v,c)$ should be as per the occurrence of color $c$ in $N(v)\cap X_i$. Moreover, $\texttt{par}(z,\alpha)$ changes for each $z\in N(v)\cap X_i$. Note that $\texttt{par}(z,c)$ remains the same for each $(z,c)\in N(v)\times([k]\setminus\{\alpha\})$. Also, $\texttt{par}(z,c)$ remains the same for each $(z,c)\in (X_i\setminus N(v))\times[k]$. Let $(j, \texttt{col}', \texttt{par}')$ be a state such that (i) $\texttt{col}'=\texttt{col}|_{X_j}$, (ii) $\texttt{par}'(z,c)=\texttt{par}(z,c)$, for each $(z,c)\in (X_i\setminus N(v))\times[k]$, (iii) $\texttt{par}'(z, \alpha)=\texttt{par}(z,\alpha)\star\texttt{odd}$, for each $z\in N(v)$, and (iv) for each $c\in[k]$, $\texttt{par}(v,c)=\Phi_{v}^c$, where $$\Phi_v^c=\begin{cases} 0, \text{ if } c \text{ does not appear in } N(v)\cap X_i;\\\texttt{odd}, \text{ if } c \text{ appears an odd number of times in } N(v)\cap X_i;\\\texttt{even}, \text{ if } c \text{ appears an even } (\ne0) \text{ number of times in } N(v)\cap X_i.\end{cases}$$

Then, we have 
$\mathcal{S}_i(i, \texttt{col}, \texttt{par}) = \mathcal{S}_j(j, \texttt{col}', \texttt{par}').$
Otherwise,
$\mathcal{S}_i(i,\texttt{col}, \texttt{par}) = 0.$

\smallskip

\textbf{Forget Node:}
Let $i$ be a forget node with child $j$ forgetting a vertex $v$. Note that $X_i=X_j\setminus\{v\}$. Consider a state $(i, \texttt{col}, \texttt{par})$. Note that $\texttt{col}$ cannot be extended to a strong odd coloring of $G$ if $\texttt{par}(v,c)\notin\{0, \texttt{odd}\}$.
Let $(j, \texttt{col}', \texttt{par}')$ be a state such that (i) $\texttt{col}'|_{X_i}=\texttt{col}$, (ii) $\texttt{par}'(z,c)=\texttt{par}(z,c)$ for each $(z,c)\in (X_j\setminus\{v\})\times[k]$, (iii) $\texttt{par}'(v,c)\in\{0,\texttt{odd}\}$ for each $c\in[k]$, and (iv) $\mathcal{S}_j(j, \texttt{col}', \texttt{par}')=1$.

Then, we have $\mathcal{S}_i(i, \texttt{col}, \texttt{par}) = 1$.
Otherwise, $\mathcal{S}(i,\texttt{col}, \texttt{par}) = 0$.

\smallskip

\textbf{Join Node:}
Let $i$ be a join node with children $j$ and $k$ such that $X_i=X_j=X_k$. Consider the states $s=(i, \texttt{col}, \texttt{par}), s'=(j, \texttt{col}', \texttt{par}')$, and $s''=(k, \texttt{col}'', \texttt{par}'')$. We say that the states $s, s'$, and $s''$ are consistent with each other if (i) $\texttt{col} = \texttt{col}' = \texttt{col}''$, and (ii) for each $(v,c)\in X_i\times[k]$, $\texttt{par}(v,c) = ((\texttt{par}'(v,c)\star\texttt{par}''(v,c))\star\Phi_v^c)$, where $$\Phi_v^c=\begin{cases} 0, \text{ if } c \text{ does not appear in } N_{X_i}(v);\\ \texttt{odd}, \text{ if } c \text{ appears an odd number of times in } N_{X_i}(v);\\ \texttt{even}, \text{ if } c \text{ appears an even number of times in } N_{X_i}(v);\end{cases}$$

Since $\texttt{col}=\texttt{col}'=\texttt{col}''$, if a color $c$ appears in $N(v)\cap X_i$ for some $v\in X_i$, then this occurrence is counted in both subtrees rooted at $j$ and $k$. Hence, directly combining $\texttt{par}'$ and $\texttt{par}''$ results in counting this occurrence twice. Therefore, we use $\Phi_v^c$ to correct this double counting and obtain the actual parity.

If there exist states $s'$ corresponding to $j$ and $s''$ corresponding to $k$ such that $\mathcal{S}(s')=\mathcal{S}(s'')=1$ and $s,s',$ and $s''$ are consistent, then $\mathcal{S}(s) = 1$. Otherwise, $\mathcal{S}(s) = 0$.

\smallskip

Now we prove the recurrence formulas used in the dynamic programming algorithm above.

\textit{Introduce node:}
We prove the correctness of the recurrence formula by showing that $\mathcal{S}_i(i,\texttt{col},\texttt{par})=1$ if and only if $\mathcal{S}_j(j,\texttt{col}',\texttt{par}')=1$.

First assume that $\mathcal{S}_i(i,\texttt{col},\texttt{par})=1$. Note that there is exactly one state, say $(j,\texttt{col}|_{X_j},\texttt{par}')$, that satisfies all the conditions. Since $\texttt{par}'$ satisfies the above-mentioned conditions, there exists a proper coloring $\varphi$ satisfying all the conditions for $\mathcal{S}_j(j,\texttt{col}|_{X_j},\texttt{par}')$ to be $1$.

Conversely, assume that $\mathcal{S}_j(j,\texttt{col}|_{X_j},\texttt{par}')=1$. Since $v$ is introduced in $X_i$ and $\texttt{par}$ satisfies the condition (iv), there exists a proper coloring $\varphi$ satisfying all the conditions for $\mathcal{S}_i(i,\texttt{col},\texttt{par})$ to be $1$. Therefore, $\mathcal{S}_i(i,\texttt{col},\texttt{par})=1$.

\smallskip

\noindent\textit{Forget node:}
Let $i$ be a forget node with child $j$ such that $X_i=X_j\setminus\{v\}$. We prove the correctness of the recurrence formula by showing that $\mathcal{S}_i(i,\texttt{col},\texttt{par})=1$ if and only if $\mathcal{S}_j(j,\texttt{col}',\texttt{par}')=1$ and $\texttt{par}(v,c)\in\{0, \texttt{odd}\}$ for each $c\in[k]$.

First, assume that $\mathcal{S}_i(i,\texttt{col},\texttt{par})=1$. Note that there $2^k$ states that satisfies all the required conditions except (iv). Although there is at most one state, say $(j,\texttt{col}|_{X_j},\texttt{par}')$, that satisfies all the conditions. Since $\texttt{col}'$ and $\texttt{par}'$ satisfies the above-mentioned conditions, there exists a proper coloring $\varphi$ satisfying all the conditions for $\mathcal{S}_j(j,\texttt{col}|_{X_j},\texttt{par}')$ to be $1$. Since we are forgetting a vertex $v$, the neighborhood of $v$ is processed completely. If $\texttt{par}(v,c)\notin\{0,1\}$ for some color $c\in[k]$, then $v$ does not satisfy strong odd condition. Therefore, $\texttt{col}$ cannot be extended to a strong odd coloring of $G$. Then, $\mathcal{S}_i(i,\texttt{col},\texttt{par})=0$. Therefore, $\texttt{par}(v,c)\in\{0, \texttt{odd}\}$

Conversely, assume that $\mathcal{S}_j(j,\texttt{col}|_{X_j},\texttt{par}')=1$ and $\texttt{par}(v,c)\in\{0, \texttt{odd}\}$ for each $c\in[k]$. Then, there exists a proper coloring $\varphi$ satisfying all the conditions for $\mathcal{S}_i(i,\texttt{col},\texttt{par})$ to be $1$. Therefore, $\mathcal{S}_i(i,\texttt{col},\texttt{par})=1$.

\smallskip
\noindent\textit{Join node:}
First, assume that $\mathcal{S}_i(i,\texttt{col},\texttt{par})=1$. Let $s=(i, \texttt{col}, \texttt{par})$. Note that there exists two states $s'=(j, \texttt{col}', \texttt{par}')$ and $s''=(k, \texttt{col}'', \texttt{par}'')$ such that $s, s'$ and $s''$ are consistent with each other. By the definition of consistency $\mathcal{S}_j(j,\texttt{col}',\texttt{par}')=1$ and $\mathcal{S}_k(k,\texttt{col}'',\texttt{par}'')=1$. 

Conversely, assume that $\mathcal{S}_j(j,\texttt{col}',\texttt{par}')=1$ and $\mathcal{S}_k(k,\texttt{col}'',\texttt{par}'')=1$. Since $\texttt{col} = \texttt{col}' = \texttt{col}''$ and for each $(v,c)\in X_i\times[k]$, $\texttt{par}(v,c) = ((\texttt{par}'(v,c)\star\texttt{par}''(v,c))\star\Phi_v^c)$, $(i, \texttt{col}, \texttt{par})$ satisfoes all the conditions for $S_i(i, \texttt{col}, \texttt{par})$ to be $1$.

\textbf{Running time:}
Note that there is exactly one state corresponding to the child of an introduce node satisfying all the required conditions and it can be searched in $O(1)$ time. Thus, leaf node and introduce node can be processed in time $O(n^{O(1)}\cdot(k.3^k)^{\texttt{tw}+1})$. Now consider a forget node. There is exactly one state corresponding to child of a forget node satisfying all the conditions required and it can be searched in $k\cdot2^k$ time. Therefore, forget node can be processed in time $O(n^{O(1)}\cdot(k.3^k)^{\texttt{tw}+1}\cdot k\cdot2^k)$. Now, consider a join node $i$ with children $j$ and $k$. Fix a coloring $\texttt{col}$ of $X_i$. We enumerate every pair $(\texttt{par}',\texttt{par}'')$ of parity functions corresponding to the two children. For each such pair, the parity function $\texttt{par}$ of the parent is uniquely determined. Since there are $k^{|X_i|}$ possible colorings of $X_i$, join node can be processed in time $O(n^{O(1)}\cdot(k\cdot9^k)^{\texttt{tw}+1})$. Therefore, the total running of the algorithm is $O(n^{O(1)}\cdot (k\cdot9^k)^{\texttt{tw}+1})$.

This shows that strong odd coloring is fixed-parameter tractable when parameterized by the treewidth of the input graph and the number of colors.
\end{proof}

Since the class of graphs with bounded treewidth is a minor-closed family, the class of graphs with bounded treewidth is contained in the class of graphs with bounded expansion. Consequently, $k\le f(\texttt{tw})$ for some function $f$~\cite{Pilipczuk2025}. 

\begin{coro}
\textsc{strong odd $k$-colorability} admits an FPT algorithm when parameterized by treewidth.
\end{coro}

Now, we prove a lower bound on \textsc{strong odd $k$-colorability} parameterized by treewidth $\texttt{tw}$.

\begin{theorem}\label{lowerbound:tw}
	For $k \geq 3$, \textsc{strong odd $k$-colorability} parameterized by treewidth $\texttt{tw}$ cannot be solved in time $(k-\varepsilon)^{\texttt{tw}}n^{O(1)}$ under SETH.
\end{theorem}

\begin{proof}
	We give a parameterized reduction from $k$-\textsc{colorability} parameterized by treewidth. 
	
\begin{construction}\label{construction} Given an instance $(G, \texttt{tw}, k)$ of $k$-\textsc{colorability}, where $\texttt{tw}$ denotes the treewidth of $G$, we construct a graph $H$ from $G$ as follows:
\begin{enumerate}
    \item Initialize $H$ as a copy of $G$. For each vertex $v \in V(G)$, let $v'$ denote the copy of $v$ in $V(H)$, and let 
    $V_1 = \{v' \mid v\in V(G)\}$.
    \item For each vertex $v' \in V_1$, attach $d_G(v) + k - 1$ pendant 
    vertices to $v'$, where $d_G(v)$ denotes the degree of $v$ in $G$. Let $P(v')=\{v_1', v_2', \ldots, v_{(d_G(v)+k-1)}'\}$ be the set of pendant vertices attached to the vertex $v'\in V_1$.
\end{enumerate}
\end{construction}

If $k\geq |V(G)|$, then \textsc{strong odd $k$-colorability} is trivially yes. So
we may assume that $k \leq |V(G)|$. Therefore, $H$ has $|V(G)|+\sum_{v\in V(G)}\bigl(d_G(v)+k-1\bigr)
=|V(G)|+2|E(G)|+|V(G)|\cdot(k-1)=2|E(G)|+|V(G)|\cdot k$ vertices, which is polynomial in the size of the input.

\begin{claim}\label{treewidth_lower_bound}
    $\texttt{tw}(H)=\texttt{tw}(G)$.
\end{claim}
\begin{proof}
Let $\mathcal{T}_G$ be a tree decomposition of $G$ with bag size at most $\texttt{tw}+1$. We construct a tree decomposition $\mathcal{T}_H$ of $H$ starting with $\mathcal{T}_G$. Since $\mathcal{T}_G$ is a tree decomposition, there exists a bag $B_v$ containing $v$ in $\mathcal{T}_G$. Now, attach $d_G(v)+k-1$ number of child bags $B_1, B_2, \ldots, B_{(d(v)+k-1)}$ to $B_v$ with $V(B_i)=\{v', v_i'\}$ for each $i\in[d_G(v)+k-1]$ where $v'$ is the copy of $v\in V(G)$. Do this for every vertex $v\in V(G)$. Then, replace $v$ with its copy $v'$ in $H$.

Now, we show that $\mathcal{T}_H$ a tree decomposition of $H$. For the copy $v'\in V(H)$ of every vertex $v\in V(G)$, we have added a bag to every edge $v'v_i'$ for each $i\in[d(v)+k-1]$. There is a bag containing $v'$ for every vertex $v'\in V(H)$, and there is a bag containing $u'$ and $v'$ for every edge $u'v'\in E(H)$. Since we have attached bags containing $B_1, B_2, \ldots, B_{(d(v)+k-1)}$ to a bag $B_v$ containing $v$ with $V(B_i)=\{v', v_i'\}$ for each $i\in[d_G(v)+k-1]$ and $v_i'$ is not contained in any other bag, bags containing $v\in V(H)$ induces a connected subtree of $\mathcal{T}_H$. Moreover, we have added bags of size at most $2$. Since $\texttt{tw}\geq1$ and bag size of $G$ is at most $\texttt{tw}+1$, $\mathcal{T}_H$ is a tree decomposition of $H$ with bag size at most $\texttt{tw}(G)+1$. Thus, $\texttt{tw}(G)=\texttt{tw}(H)$.
\end{proof}

\begin{claim}\label{cons:iff}
    $G$ is $k$-colorable if and only if $H$ is strong odd $k$-colorable.
\end{claim}
\begin{proof}
Assume that $\psi : V(G) \to [k]$ is a proper $k$-coloring of $G$. We define a coloring $\varphi$ for $H$. For each vertex $v'\in V_1$, $\varphi(v') = \psi(v)$. For each vertex $v \in V(G)$, let $N(v) = \{u_1, \dots, u_{d_G(v)}\}$. For each $v\in V_1$, $\varphi(v_i') = \psi(u_i)$ for $i\in[d_G(v)]$. For each $v\in V_1$, assign distinct colors from the set $[k] \setminus \{\varphi(v)\}$ to the vertices of 
$\{v_{d_G(v)+1}', v_{d_G(v)+2}' \ldots, v_{d_G(v)+k-1}'\}$, one color per vertex. 

\smallskip

We claim that $\varphi$ is a strong odd coloring of $H$. Since $\psi$ is a proper coloring of $G$ and we do not assign $\varphi(v)$ to any vertex from the set $P(v)$, $\varphi$ is a strong odd coloring of $H$. Now, consider any vertex $v\in V(G)$. Let $c$ be a color that appears in $N_H(v)$. If $c$ appears $m$ number of times in $N(v)$, then $c$ appears $m$ number of times in $\{v_1, v_2, \ldots, v_{d_G(v)}\}$ and once in $\{v_{d_G(v)+1}, v_{d_G(v)+2}, \ldots, v_{d_G(v)+k-1}\}$. Therefore, $c$ appears an odd number of times in $N_H(v)$. Now, assume that $c$ does not appear in $N(v)$. The color $c$ appears exactly once in $N_H(v)$. Therefore, each color appears an odd number of times in $N_H(v)$. Since each vertex $v\in V(H)\setminus V(G)$ is a vertex of degree 1, every color appearing in $N_H(v)$ appears an odd number of times.
Thus, $\varphi$ is a strong odd $k$-coloring of $H$.

Let $\varphi$ be a strong odd $k$ coloring of $H$. Since $G$ is an induced subgraph of $H$, restriction of $\varphi$ to $G$ is a proper $k$-coloring of $G$.
\end{proof}

Lokshtanov et al.~\cite{Lokshtanov2018} have shown that $k$\textsc{-Colorability} cannot be solved in time $(k-\varepsilon)^{\texttt{tw}(G)}n^{O(1)}$ time under SETH, for any fixed $k\ge3$. Therefore, \textsc{strong odd $k$-colorability} parameterized by treewidth $\texttt{tw}$ cannot be solved in time $(k-\varepsilon)^{\texttt{tw}}n^{O(1)}$ under SETH, for $k \geq 3$.
\qed\end{proof}

Now, we study \textsc{strong odd $k$-colorability} parameterized by the size of a feedback vertex set. We use the following kernelization lower bound for $k$-\textsc{colorability}, proved by Bodlaender et al.~\cite{BJK}.

\begin{theorem}[Bodlaender et al.~\cite{BJK}]\label{graph_col_vc}
$k$-\textsc{colorability} parameterized by the size of a vertex cover does not admit a polynomial kernel unless $\mathrm{NP}\subseteq\mathrm{coNP}/\mathrm{poly}$.
\end{theorem}

Note that every vertex cover of a graph is also a feedback vertex set. Therefore, for every graph $G$, $\texttt{fvs}(G)\leq\texttt{vc}(G)$. Therefore, $k$-\textsc{colorability} parameterized by the size of a feedback vertex set does not admit a polynomial kernel unless $\mathrm{NP}\subseteq\mathrm{coNP}/\mathrm{poly}$.

\begin{theorem}
\textsc{strong odd $k$-colorability} parameterized by the size of a feedback vertex set does not admit a polynomial kernel unless $\mathrm{NP}\subseteq\mathrm{coNP}/\mathrm{poly}$.
\end{theorem}
\begin{proof}
We prove the result by giving a polynomial-parameter transformation from $k$-\textsc{colorability} parameterized by feedback vertex set number to \textsc{strong odd $k$-colorability} with the same parameter. Let $(G, X_G, d, k)$ be an instance of $k$-\textsc{colorability}, where $X_G$ is a feedback vertex set of size $d$. We construct $H$ from $G$ as  constructed in Construction~\ref{construction} in the proof of Theorem~\ref{lowerbound:tw}.

Let $X_H := X_G$. Since $G-X_G$ is a forest and all vertices added to $G$ while constructing $H$ are pendant vertices, $H-X_H$ is a forest. Thus, $X_H$ is a feedback vertex set of $H$ with $|X_H|=|X_G|=d$.

By Claim~\ref{cons:iff}, $G$ is $k$-colorable if and only if $H$ is strong odd $k$-colorable. Therefore, \textsc{strong odd $k$-colorability} parameterized by the size of a feedback vertex set does not admit a polynomial kernel unless $\mathrm{NP}\subseteq\mathrm{coNP}/\mathrm{poly}$.
\qed\end{proof}

Now, we study the parameterized complexity of \textsc{strong odd $k$-colorability} when parameterized by clique-width of the input graph. Fomin et al.~\cite{Fomin} proved the following.

\begin{theorem}[Fomin et al.~\cite{Fomin}]
$k$-\textsc{colorability} is $\mathrm{W}[1]$-hard when parameterized by clique-width.
\end{theorem}

We transfer this hardness result to \textsc{strong odd $k$-colorability} by a parameter-preserving reduction from $k$-\textsc{colorability}.

\begin{theorem}
\textsc{strong odd $k$-colorability} is $\mathrm{W}[1]$-hard when parameterized by clique-width.
\end{theorem}

\begin{proof}
We give a parameterized reduction from $k$-\textsc{colorability} parameterized by clique-width. Given an instance $(G, w, k)$ of $k$-\textsc{colorability}, where $w$ denotes the clique-width of $G$, we construct a graph $H$ from $G$ as constructed in the proof of Theorem~\ref{lowerbound:tw}.

\begin{claim}\label{cliquewidth}
	$cw(H)\leq cw(G)+3$.
\end{claim}
\begin{proof}
	Let $\Phi_G$ be a $w$-expression of $G$. Let $v\in V(G)$ be a vertex that is introduced with label $i$ in
	$\Phi_G$. We replace the subexpression $v(i)$ by 
	\[\left(
	\rho_{w+2\to w+3}
	\left(
	\rho_{w+1\to i}
	\left(
	\eta_{w+1,w+2}
	\left(
	v'(w+1)
	\oplus v_1'(w+2)
	\oplus v_2'(w+2)
	\oplus \ldots 
	\oplus\\ v_{d_G(v)+k-1}'(w+2)
	\right)
	\right)
	\right)
	\right),\]
	where $v'$ is the copy of $v$ in $H$ and
	$v'_1,v'_2,\ldots,v'_{d_G(v)+k-1}$ are the pendant vertices attached to $v'$.
	
	Initially, $v'$ receives label $w+1$, while all pendant vertices in
	$P(v')$ receive label $w+2$. The operation
	$\eta_{w+1,w+2}$ makes every vertex in $P(v')$ adjacent to $v'$. We then rename $w+1$ to the original label $i$ of $v$ and rename $w+2$ to $w+3$. Since label $w+3$ is not involved in any introducing edge operation, no edges are introduced between the vertices in $P(v')$ and any vertex except $v'$. Thus, every vertex in $P(v')$ is adjacent only to $v'$. Renaming $w+1$ to the original label $i$ will create all the edges in $G$ by introducing edge operations of $\Phi_G$ on the copies of the vertices of $G$ in $H$. Hence, the resulting expression constructs $H$ and uses at most $w+3$ labels. Therefore, $\texttt{cw}(H)\leq w+3$.
\end{proof}

By Claim~\ref{cons:iff}, $G$ is $k$-colorable if and only if $H$ is strong odd $k$-colorable. Therefore, \textsc{strong odd $k$-colorability} is $\mathrm{W}[1]$-hard when parameterized by clique-width.
\end{proof}


\begin{thebibliography}{99}

\bibitem{AHU_textbook}
A. V. Aho, J. E. Hopcroft, J. D. Ullman, The Design and Analysis of Computer Algorithms,
Addison–Wesley, Boston, 1974.



\bibitem{ACCCHHHKZ}
J. Anderson, H. Chau, E.-K. Cho, N. Crawford, S. G. Hartke, E. Heath, O. Henderschedt, H. Kwon,
Z. Zhang, The forb-flex method for odd coloring and proper coloring of planar graphs, Discrete
Mathematics 348 (2025) 114648.



\bibitem{Belmonte2026}
R. Belmonte, J. P. Bravo, N. Köhler, H. Müller, Complexity classification of colouring problems
with parity constraints (2026). 


\bibitem{Bhyravarapu2026}
S. Bhyravarapu, S. Kumari, I. V. Reddy, On the parameterized complexity of odd coloring, Theoretical Computer Science 1082 (2026) 116074.


\bibitem{BJK}
H. L. Bodlaender, B. M. P. Jansen, S. Kratsch, Cross-composition: A new technique for kernelization lower bounds, in: Proceedings of the 28th International Symposium on Theoretical Aspects
of Computer Science (STACS), 2011, pp. 165–176.


\bibitem{textBook}
J. A. Bondy, U. S. R. Murty, Graph Theory, Graduate Texts in Mathematics, Springer, New York,
2008.



\bibitem{CPS_odd}
Y. Caro, M. Petruševski, R. Škrekovski, Remarks on odd colorings of graphs, Discrete Applied
Mathematics 321 (2022) 392–401.


\bibitem{Caro2025}
Y. Caro, M. Petruševski, R. Škrekovski, Z. Tuza, On strong odd colorings of graphs, Discrete
Mathematics 348 (11) (2025) 114601. 


\bibitem{Cho2023}
E.-K. Cho, I. Choi, H. Kwon, B. Park, Odd coloring of sparse graphs and planar graphs, Discrete
Mathematics 346 (5) (2023) 113305.



\bibitem{Cygan2015}
M. Cygan, F. V. Fomin, L. Kowalik, D. Lokshtanov, D. Marx, M. Pilipczuk, M. Pilipczuk, S. Saurabh, Parameterized Algorithms, Springer International Publishing, 2015. 




\bibitem{Fomin}
F. V. Fomin, P. A. Golovach, D. Lokshtanov, S. Saurabh, Intractability of clique-width parame-
terizations, SIAM Journal on Computing 39 (5) (2010) 1941–1956.




\bibitem{Goetze2025}
M. Goetze, F. Klute, K. Knauer, I. Parada, J. P. Peña, T. Ueckerdt, Strong odd coloring in minor-closed classes (2025). 




\bibitem{KMOOT}
M. Kashima, S. Maezawa, K. Osako, K. Ozeki, S. Tsuchiya, An odd 4-coloring of a maximal
outerplanar graph, Discrete Mathematics 348 (2025) 114556.


\bibitem{Kashima2026}
M. Kashima, K. Ozeki, Odd coloring of k-trees, Discrete Applied Mathematics 390 (2026) 88–101.



\bibitem{Kwon2026}
H. Kwon, B. Park, Strong odd coloring of sparse graphs, Discrete Mathematics 349 (11) (2026)
115292.




\bibitem{Lokshtanov2018}
D. Lokshtanov, D. Marx, S. Saurabh, Known algorithms on graphs of bounded treewidth are probably
optimal, ACM Transactions on Algorithms 14 (2) (2018) 1–30. 


\bibitem{Manattu2026}
A. J. Manattu, A. Vinay, A. Lakshmanan S, On a conjecture about the strong odd chromatic
number of planar graphs, Discrete Mathematics 349 (12) (2026) 115293. 





\bibitem{Pang2026}
J.-R. Pang, L.-Y. Miao, Y.-Z. Fan, On odd and strong odd colorings of graphs, Discrete Mathe-
matics 349 (1) (2026) 114683. 



\bibitem{PS1}
M. Petruševski, R. Škrekovski, Colorings with neighborhood parity condition, Discrete Applied
Mathematics 321 (2022) 385–391.


\bibitem{Pilipczuk2025}
M. Pilipczuk, Strong odd colorings in graph classes of bounded expansion, The Electronic Journal
of Combinatorics 32 (4) (Oct. 2025). 



\bibitem{Zuckerman}
D. Zuckerman, Linear degree extractors and the inapproximability of max clique and chromatic
number, Theory of Computing 3 (2007) 103–128.


\end{thebibliography}
\end{document}